\documentclass[letterpaper, 10 pt, conference]{ieeeconf}%
\usepackage{graphicx}
\usepackage{epsfig}
\usepackage{times}
\usepackage{amsmath}
\usepackage{thmtools,thm-restate}
\usepackage{amssymb}
\usepackage[section]{placeins}
\usepackage{placeins}
\usepackage{amsmath}
\usepackage{multirow}
\usepackage{graphicx}
\usepackage[normalem]{ulem}
\usepackage{subcaption}
\usepackage{algorithm,tabularx}
\usepackage{algpseudocode}
\usepackage{amsfonts}
\usepackage{cases}
\usepackage{romannum}
\usepackage{textcomp}
\usepackage{xcolor}%
\usepackage{textcomp}
\providecommand{\U}[1]{\protect\rule{.1in}{.1in}}
\providecommand{\U}[1]{\protect\rule{.1in}{.1in}}

\IEEEoverridecommandlockouts
\newcommand{\R}{\mathbb{R}}

\newcommand{\Z}{\mathbb{Z}}
\newcommand{\T}{\mathsf{T}}
\newcommand{\I}{\mathbf{I}}
\newcommand{\0}{\mathbf{0}}

\newcommand{\tsup}[1]{\textsuperscript{#1}}

\newtheorem{assumption}{Assumption}
\newtheorem{theorem}{Theorem}
\newtheorem{corollary}{Corollary}
\newtheorem{lemma}{Lemma}

\newtheorem{remark}{Remark}
\newtheorem{definition}{Definition}

\useunder{\uline}{\ul}{}
\makeatletter
\newcommand{\multiline}[1]{  \begin{tabularx}{\dimexpr\linewidth-\ALG@thistlm}[t]{@{}X@{}}
#1
\end{tabularx}
}
\makeatother
\usepackage{cite}
\usepackage{hyperref}
\usepackage{picins}
\usepackage{centernot}
\let\labelindent\relax
\usepackage{enumitem}
\setlist[itemize]{leftmargin=*}
\usepackage{makecell}
\usepackage[normalem]{ulem}
\usepackage{cuted}

\algdef{SE}[SUBALG]{Indent}{EndIndent}{}{\algorithmicend\ }%
\algtext*{Indent}
\algtext*{EndIndent}

\title{\LARGE \bf
On-line Estimation of Stability and Passivity Metrics
}

\author{Shirantha Welikala, Hai Lin and Panos J. Antsaklis 
\thanks{The support of the National Science Foundation (Grant No. IIS-1724070, CNS-1830335, IIS-2007949) is gratefully acknowledged.}
\thanks{The authors are with the Department of Electrical Engineering, College of Engineering, University of Notre Dame, IN 46556, \texttt{{\small \{wwelikal,hlin1,pantsakl\}@nd.edu}}.}}

\begin{document}

\maketitle
\thispagestyle{empty}
\pagestyle{empty}


\begin{abstract}
We consider the problem of on-line evaluation of critical characteristic parameters such as the $\mathcal{L}_2$-gain (L2G), input feedforward passivity index (IFP) and output feedback passivity index (OFP) of non-linear systems using their input-output data. Typically, having an accurate measure of such \emph{system indices} enables the application of systematic control design techniques. 
Moreover, if such system indices can efficiently be evaluated on-line, they can be exploited to device intelligent controller reconfiguration and fault-tolerant control techniques. 
However, the existing estimation methods of such system indices (i.e., L2G, IFP and OFP) are predominantly off-line, computationally inefficient, and require a large amount of actual or synthetically generated input-output trajectory data under some specific initial/terminal conditions. 
On the other hand, the existing on-line estimation methods take an averaging-based approach, which may be sub-optimal, computationally inefficient and susceptible to estimate saturation. 
In this paper, to overcome these challenges (in the on-line estimation of system indices), we establish and exploit several interesting theoretical results on a particular class of fractional function optimization problems. For comparison purposes, the details of an existing averaging-based approach are provided for the same on-line estimation problem. Finally, several numerical examples are discussed to demonstrate the proposed on-line estimation approach and to highlight our contributions.

\end{abstract}

\section{Introduction}\label{Sec:Introduction}


Stability and passivity are two critical concepts that are being abundantly used in various control systems design tasks \cite{Khalil1996,Bao2007,Willems1972a}. In particular, quantitative measures of stability and passivity such as the $\mathcal{L}_2$-gain (L2G), input-feedforward passivity index (IFP) and output feedback passivity index (OFP) provide convenient avenues for control systems design \cite{Zames1966,Desoer1975}. This is because such quantitative measures of a system can adequately characterize the system in lieu of an accurate theoretical model \cite{Tang2019}. Moreover, with the increasing complexity of systems, the problem of identifying an accurate system model becomes extremely challenging \cite{Koch2021}. 
Thus, for such instances, designing control solutions based on estimated quantitative measures like L2G, IFP and OFP (henceforth, collectively referred to as the ``\emph{system indices}'') is more suitable than designing control solutions based on estimated system models  \cite{Tanemura2019b}. 
Along the same lines, having an accurate and on-line estimate of such system indices paves the way to device intelligent controller reconfiguration and fault-tolerant control techniques. 
Therefore, this paper focuses on the problem of estimating system indices on-line using input-output data of the system.

Different approaches have been used in the literature to address this estimation problem under different settings. In \cite{Romer2017a}, certain dissipation inequalities (that provide bounds on system indices) are determined using a finite set of input-output data samples (tuples). While this approach is valid for non-linear systems, it is computationally inefficient, off-line and may require many data samples to get reasonable bounds. In contrast, the work in \cite{Romer2017b} and its extension \cite{Koch2021} have proposed a systematic input-output data sample generation method for the exact evaluation of system indices. However, this approach assumes the ability to generate custom input-output data samples from the system - making it an off-line solution. Moreover, it is limited to linear time-invariant systems and is computationally inefficient.

The subsequent work in \cite{Tanemura2019} has improved the computational efficiency of the solution proposed in \cite{Romer2017b} by two folds via proving that only half of the input-output data samples required in \cite{Romer2017b} is sufficient to obtain the same final solution. The work in \cite{Iijima2020} has further enhanced the computational efficiency by improving the used gradient-based update scheme to achieve faster convergence. However the solution proposed in \cite{Iijima2020} suffers from the remaining drawbacks mentioned earlier (for \cite{Romer2017b,Koch2021}). 

Focusing on non-linear systems, the work in \cite{Tang2019} and its extension \cite{Tang2021} has developed a dissipativity learning and control framework inspired by machine learning techniques such as one-class support vector machines. However, this approach is an off-line method that is computationally inefficient and requires many input-output data samples. 

On the other hand, in \cite{Zakeri2019}, an averaging-based simple on-line approach for the estimation of system indices has been proposed. This approach is applicable for non-linear systems as well. However, its estimates are often sub-optimal, computationally expensive and susceptible to saturation effects over time. In this paper, we continue this line of research and develop an on-line estimation method for system indices while addressing the said challenges of \cite{Zakeri2019}.

\paragraph{\textbf{Contributions}}
Our main contributions are as follows: 
(i) We consider a general continuous-time non-linear system and formulate input-output data-driven on-line estimation problems to determine the L2G, IFP and OFP (i.e., the system indices) - each as a fractional function optimization problem (FFOP) over a two-dimensional space; 
(ii) For a generalized class of these FFOPs, we develop an efficient solution approach that converts an FFOP to a much simpler optimization problem over a one-dimensional space; (iii) We further prove the validity of the proposed FFOP solution approach in the presence of ``dead-zones'' (where the denominators of the fractional functions become zero);  (iv) We discuss how optimal control can be exploited to control the system to rapidly achieve better estimates for the system indices (rather than passively using the observed input-output data); (v) We provide theoretical results that characterize the relative locations of the optimal system indices with respect to the existing and proposed on-line estimates of such system indices; (vi) We extend the proposed solution for general discrete-time non-linear systems; (vii) Finally, we provide several numerical examples to demonstrate the validity and the potential of the proposed approach.

\paragraph{\textbf{Organization}} This paper is organized as follows. The preliminary concepts behind the characteristic system indices of interest are introduced in Sec. \ref{Sec:Preliminaries}. The formulation of the interested system indices estimation problems is provided in Sec. \ref{Sec:ProblemFormulation}. In Sec. \ref{Sec:FractionalFunctionOptimization}, a particular class of useful FFOPs is analyzed and the details of the proposed FFOP based on-line estimation method are provided. In Sec.   \ref{Sec:GeneralizationToDiscreteTimeSystems}, the proposed problem formulation and the corresponding solution is specialized for discrete-time systems. Finally, several numerical results are summarized in Sec. \ref{Sec:NumericalResults} before concluding the paper in Sec. \ref{Sec:Conclusion}.

\paragraph{\textbf{Notation}}
The sets of real and integer numbers are denoted by $\R$ and $\Z$, respectively. We use the notations $[\cdot,\cdot]$ and $[\cdot,\cdot)$ to represent closed and left-closed right-open intervals, respectively. These intervals may be either continuous (in $\R$) or discrete (in $\Z$) - which will be clear from the context. The notation $(\cdot,\cdot)$ is used to represent $2$-tuples. Given sets $A$ and $B$, $A\backslash B$ indicates the set subtraction operation. The notation $\vert \cdot \vert$ represents the cardinality (or absolute) operation if the argument is a set (or number). The transpose of a matrix $A$ is denoted by $A^\T$. The zero and identity matrices are denoted by $\0$ and $\I$ (dimensions will be clear from the context), respectively.


\section{Preliminaries}\label{Sec:Preliminaries}

\paragraph{\textbf{System}} Consider the non-linear system:
\begin{equation}\label{Eq:NonlinearSystem}
    \mathcal{H}:
    \begin{cases}
    \dot{x}(t) = f(x(t),u(t)),\\
    y(t) = h(x(t),u(t)),
    \end{cases}
\end{equation}
where $x(t) \in X \subset \R^n$, $u(t) \in U \subset \R^m$ and $y(t) \in Y \subset \R^m$ are state, input and output variables at time $t\in\R_{\geq 0}$, and $X, U$ and $Y$ are state, input and output spaces, respectively. Let us use the notation $x(t) \triangleq \phi(t,t_0,x_0,u_{t_0:t})$ to denote the state at time $t$ reached from the initial state $x_0$ at initial time $t_0$ under the control profile $u_{t_0:t} \triangleq \{u(\tau),\tau\in[t_0,t]\}$. 

By definition, the following assumption holds regarding the dynamics and the feasible spaces $X,U,Y$ given in \eqref{Eq:NonlinearSystem}.
\begin{assumption} \label{As:Dynamics}
For any $x_0 \in X$ and $u_{t_0:t} \subset U$, the system $\mathcal{H}$ \eqref{Eq:NonlinearSystem} results in $x(t) = \phi(t,t_0,x_0,u_{t:t_0}) \in X$ and $y(t)=h(\phi(t,t_0,x_0,u_{t_0:t}),u(t)) \in Y$, for all $t\geq t_0 \geq 0$. 
\end{assumption}

\paragraph{\textbf{Supply Rate, Storage Function and Dissipativity}}
A \emph{supply rate} $w(t) = w(u(t),y(t))$ represents a considered rate of energy inserted into a system. In particular, it satisfies the condition stated in the following Def. \ref{Def:SupplyRate} and helps to define the \emph{dissipativity} of a system in the subsequent Def. \ref{Def:Dissipativity}.    

\begin{definition}\label{Def:SupplyRate} \cite{Bao2007} 
The \emph{supply rate} $w(t)=w(u(t),y(t))$ for the system $\mathcal{H}$ \eqref{Eq:NonlinearSystem} is such that, for any $x_0\in X$ and $u_{t_0:t} \subset U$ with $y(t) = h(\phi(t,t_0,x_0,u_{t_0:t}), u(t)) \in Y$, it satisfies 
$
    \int_{t_0}^{t_1} \vert w(t) \vert dt \leq \infty,
    \mbox{ for all } t_1\geq t_0 \geq 0.
$
\end{definition}

\begin{definition}\label{Def:Dissipativity} \cite{Bao2007} 
The system $\mathcal{H}$ \eqref{Eq:NonlinearSystem} under the supply rate $w(t)=w(u(t),y(t))$ is said to be \emph{dissipative} if there exists a positive semidefinite function $S(x):X\rightarrow\R_{\geq0}$ called the \emph{storage function}, such that for any $x_0\in X$ and $u_{t_0:t} \subset U$, 
\begin{equation}
S(x_1) - S(x_0) \leq \int_{t_0}^{t_1}w(t)dt,    
\end{equation}
for all $t_1\geq t_0 \geq 0$, where $x_1 \triangleq x(t_1) = \phi(t_1,t_0,x_0,u_{t_0:t_1})$.
\end{definition}

Note that the \emph{storage function} $S(x)$ mentioned in Def. \ref{Def:Dissipativity} represents the energy of the system at a state $x$. 

Next, we define the \emph{available storage} $S_a(x):X\rightarrow\R_{\geq 0}$, which is the largest amount of energy that can be extracted from a system starting from an initial condition $x(0) = x$.
\begin{definition}\label{Def:AvailableStorage}\cite{Bao2007} 
The \emph{available storage} $S_a(x)$ of the system $\mathcal{H}$ \eqref{Eq:NonlinearSystem} with the supply rate $w(t)=w(u(t),y(t))$ is:
\begin{equation}
    S_a(x) \triangleq \underset{\substack{x(0)=x,\,u_{0:t} \subset U,\,t_1>0}}{\sup}\ \left\{-\int_0^{t_1} w(t) \,dt \right\}.
\end{equation}
\end{definition}

As shown in \cite{Bao2007}: 
(i) any possible storage function $S(x)$ satisfies $0 \leq S_a(x) \leq S(x), \forall x\in X$, 
(ii) if $S_a(x)$ is continuous, $S_a(x)$ itself is a possible storage function, and 
(iii) if $S_a(x)<\infty, \forall x\in X$, then the system is dissipative with respect to the corresponding supply rate $w(t)$.

\paragraph{\textbf{System Indices}} In the following Defs. \ref{Def:L2Gain}-\ref{Def:OFPIndex}, we define the concepts of: (i) finite-gain $\mathcal{L}_2$ stability, (ii) input feedforward passivity and (iii) output feedback passivity, along with their corresponding indices. Note that, henceforth, we use the notation $x_1 \triangleq x(t_1) = \phi(t_1,t_0,x_0,u_{t_0:t_1})$.

\begin{definition}\label{Def:L2Gain}\cite{Bao2007}
The system $\mathcal{H}$ \eqref{Eq:NonlinearSystem} is \emph{finite-gain $\mathcal{L}_2$ stable} with the gain $\gamma$ (also denoted as: $\mathcal{H}$ is L2G($\gamma$)), if there exists a storage function $S:X\rightarrow \R_{\geq 0}$ with $S(\0)=0$ such that for any $x_0\in X$ and $u_{t_0:t} \subset U$, 
\begin{equation}\label{Eq:Def:L2Gain}
S(x_1) - S(x_0) \leq \int_{t_0}^{t_1} \gamma^2 u^\T(t)u(t) - y^\T(t)y(t) dt,
\end{equation}
for all $t_1 \geq t_0 \geq 0$.
\end{definition}

Note that, if the system $\mathcal{H}$ \eqref{Eq:NonlinearSystem} is linear and is characterized by a transfer function $H(s)$, then, its $\mathcal{L}_2$-gain is given by $\gamma = \sup_{\omega\in\R} \Vert H(j\omega) \Vert_2$, where $j$ represents the \emph{imaginary unit}  \cite{Khalil1996}. Therefore, if $\gamma<1$, then, $\mathcal{H}$ has a positive gain margin and (consequently) is stable.

The knowledge of the $\mathcal{L}_2$-gain value of a system is not only important to reach conclusions regarding its stability, but also can be helpful to synthesize controllers (e.g., via small-gain theorem \cite{Bao2007}) and controller re-configurations \cite{Zakeri2019}.

\begin{definition}\label{Def:Passivity} \cite{Bao2007}
The system $\mathcal{H}$ \eqref{Eq:NonlinearSystem} is \emph{passive} if there exists a storage function $S:X\rightarrow \R_{\geq 0}$ with $S(\0)=0$ such that for any $x_0\in X$ and $u_{t_0:t} \subset U$,
\begin{equation}
S(x_1) - S(x_0) \leq \int_{t_0}^{t_1}u^\T(t)y(t)dt,
\end{equation}
for all $t_1 \geq t_0 \geq 0$.
\end{definition}

The concept of passivity strives to provide an energy-based characterization to the input-output behavior of a dynamical system. In essence, a passive system can be thought of as a system that stores and dissipates energy without generating its own \cite{Willems1972a,Kottenstette2014}. Under some additional assumptions, passivity implies stability. Moreover, it can be conveniently used to device controllers \cite{Bao2007} and controller re-configurations \cite{Zakeri2019}. Motivated by these potential usages, we next define two degrees of passivity that can be characterized corresponding to a dynamical system.

\begin{definition}\label{Def:IFPIndex}
The system $\mathcal{H}$ \eqref{Eq:NonlinearSystem} is \emph{input feedforward passive} with the index $\nu\in\R$ (also denoted as: $\mathcal{H}$ is IFP($\nu$)), if there exist a storage function $S:X\rightarrow \R_{\geq 0}$ with $S(\0)=0$ such that for any $x_0\in X$ and $u_{t_0:t} \subset U$,
\begin{equation}\label{Eq:Def:IFPIndex}   
S(x_1) - S(x_0) \leq \int_{t_0}^{t_1} u^\T(t)y(t) - \nu u^\T(t)u(t) dt.
\end{equation}
for all $t_1 \geq t_0 \geq 0$. Moreover, if $\mathcal{H}$ is IFP($\nu$) with $\nu>0$, then, $\mathcal{H}$ is said to be \emph{input strictly passive}. 
\end{definition}

\begin{definition}\label{Def:OFPIndex}\cite{Bao2007}
The system $\mathcal{H}$ \eqref{Eq:NonlinearSystem} is \emph{output feedback passive} with the index $\rho\in\R$ (also denoted as: $\mathcal{H}$ is OFP($\rho$)), if there exist a storage function $S:X\rightarrow \R_{\geq 0}$ with $S(\0)=0$ such that for any $x_0\in X$ and $u_{t_0:t} \subset U$,
\begin{equation}\label{Eq:Def:OFPIndex} 
S(x_1) - S(x_0) \leq \int_{t_0}^{t_1} u^\T(t)y(t) - \rho y^\T(t)y(t) dt, 
\end{equation}
for all $t_1 \geq t_0 \geq 0$. Moreover, if $\mathcal{H}$ is OFP($\rho$) with $\rho>0$, then, $\mathcal{H}$ is said to be \emph{output strictly passive}.
\end{definition}

Based on Defs. \ref{Def:L2Gain} and \ref{Def:OFPIndex}, it can be shown that if the system $\mathcal{H}$ \eqref{Eq:NonlinearSystem} is OFP($\rho$) with $\rho>0$, then it is L2G($\gamma$) with $\gamma\leq \frac{1}{\rho}$.

It is worth noting that each of the properties defined in Defs. \ref{Def:L2Gain}-\ref{Def:OFPIndex} can be interpreted as the dissipativity property (Def. \ref{Def:Dissipativity}) under a specific supply rate. For example, note that, $\mathcal{H}$ is L2G($\gamma$) if and only if $\mathcal{H}$ is dissipative under the supply rate $w(u(t),y(t)) = \gamma^2 u^\T(t)u(t) - y^\T(t)y(t)$. 

At this point, we remind that our main focus in this paper is on estimating the \emph{system indices}: (i) $\mathcal{L}_2$-gain (L2G) $\gamma^2$, (ii) input feedforward passivity (IFP) $\nu$ and (iii) output feedback passivity (OFP) $\rho$, of the system $\mathcal{H}$ \eqref{Eq:NonlinearSystem}.

\begin{remark}\label{Rm:ActualSystemIndices}
Based on the forms of \eqref{Eq:Def:L2Gain}, \eqref{Eq:Def:IFPIndex} and \eqref{Eq:Def:OFPIndex}, it is clear that if the system $\mathcal{H}$ \eqref{Eq:NonlinearSystem} is L2G($\gamma$), IFP($\nu$) and OFP($\rho$), then, it is also L2G($\gamma+\epsilon$), IFP($\nu-\epsilon$) and OFP($\rho-\epsilon$), respectively, for all $\epsilon>0$. Therefore, the smallest possible $\gamma$ value that satisfies \eqref{Eq:Def:L2Gain} in Def. \ref{Def:L2Gain}, is the actual/optimal $L_2$-gain of the system. Similarly, the largest possible $\nu$ and $\rho$ values that respectively satisfy \eqref{Eq:Def:IFPIndex} and \eqref{Eq:Def:OFPIndex} in Defs. \ref{Def:IFPIndex} and \ref{Def:OFPIndex}, are the actual/optimal IFP and OFP indices of the system. 
\end{remark}

Taking Rm. \ref{Rm:ActualSystemIndices} into account, let us denote the \emph{optimal system indices} L2G, IFP and OFP as $\gamma_*, \nu_*$ and $\rho_*$, respectively. In the sequel, we will show how these optimal system indices can be estimated.

\section{Problem Formulation}\label{Sec:ProblemFormulation}

\paragraph{\textbf{Optimal System Indices}} Let us first define three fractional functions as 
\begin{equation}\label{Eq:FractionalFunctionsOfInterest}
\begin{aligned}
    \psi_{\gamma}(t_1,t_0,x_0,u_{t_0:t_1}) \triangleq 
    \frac{\int_{t_0}^{t_1} y^\T(t)y(t)dt + (S(x_1)-S(x_0))}{\int_{t_0}^{t_1}u^\T(t)u(t)dt},\\
    \psi_{\nu}(t_1,t_0,x_0,u_{t_0:t_1}) \triangleq 
    \frac{\int_{t_0}^{t_1} u^\T(t)y(t)dt - (S(x_1)-S(x_0))}{\int_{t_0}^{t_1}u^\T(t)u(t)dt},\\
    \psi_{\rho}(t_1,t_0,x_0,u_{t_0:t_1}) \triangleq 
    \frac{\int_{t_0}^{t_1} u^\T(t)y(t)dt - (S(x_1)-S(x_0))}{\int_{t_0}^{t_1}y^\T(t)y(t)dt},\\
\end{aligned}
\end{equation}
respectively over the spaces $\Psi_\gamma \triangleq \Psi \backslash \Psi_u$, $\Psi_\nu \triangleq \Psi \backslash \Psi_u$ and $\Psi_\rho \triangleq \Psi \backslash \Psi_y$, where 
\begin{equation*}\label{Eq:FractionalFunctionsOfInterestSpaces}
\begin{aligned}
    \Psi\ \, \triangleq&\ \{(t_0,t_1,x_0,u_{t_0:t_1}): t_1 > t_0 \geq 0, x_0\in X, u_{t_0:t_1}\subset U\},\\
    \Psi_u \triangleq&\ \{(t_0,t_1,x_0,u_{t_0:t_1})\in\Psi:\int_{t_0}^{t_1}u^\T(t)u(t)dt = 0\},\\
    \Psi_y \triangleq&\ \{(t_0,t_1,x_0,u_{t_0:t_1})\in\Psi:\int_{t_0}^{t_1}y^\T(t)y(t)dt = 0\}.
\end{aligned}
\end{equation*}

Using \eqref{Eq:FractionalFunctionsOfInterest} and Rm. \ref{Rm:ActualSystemIndices}, the following lemma expresses the optimal system indices as optimization problems.

\begin{lemma}\label{Lm:OptimalSystemIndices}
The optimal system indices: L2G $\gamma_*$, IFP $\nu_*$ and OFP $\rho_*$ of the system $\mathcal{H}$ \eqref{Eq:NonlinearSystem} are given respectively by
\begin{equation}\label{Eq:Lm:OptimalSystemIndices}
\begin{aligned}
    \gamma^2_* = \max_{(t_1,t_0,x_0,u_{t_0:t_1})\in\Psi_\gamma}\ \psi_{\gamma}(t_1,t_0,x_0,u_{t_0:t_1}), \\
    \nu_* = \min_{(t_1,t_0,x_0,u_{t_0:t_1})\in\Psi_\nu}\ \psi_{\nu}(t_1,t_0,x_0,u_{t_0:t_1}), \\
    \rho_* = \min_{(t_1,t_0,x_0,u_{t_0:t_1})\in\Psi_\rho}\ \psi_{\rho}(t_1,t_0,x_0,u_{t_0:t_1}).
\end{aligned}
\end{equation}
\end{lemma}
\begin{proof}
Based Rm. \ref{Rm:ActualSystemIndices}, by re-arranging \eqref{Eq:Def:L2Gain} and substituting $\psi_\gamma$ from \eqref{Eq:FractionalFunctionsOfInterest}, the optimal L2G value $\gamma_*^2$ can be expressed as
\begin{equation}\label{Eq:Lm:OptimalSystemIndicesStep1}
\begin{aligned}
    \gamma_*^2 =&\ \underset{\gamma}{\max}\ &&\gamma^2 \\
    & \mbox{sub. to: } &&\psi_{\gamma}(t_1,t_0,x_0,u_{t_0:t_1}) \leq \gamma^2, \\
    & && \ \ \ \ \ \ \ \forall (t_1,t_0,x_0,u_{t_0:t_1})\in\Psi_\gamma.
\end{aligned}
\end{equation}
Note that the optimization problem in \eqref{Eq:Lm:OptimalSystemIndicesStep1} can be equivalently represented by the first problem (i.e., the $\gamma_*^2$ expression) in \eqref{Eq:Lm:OptimalSystemIndices}. Following the same steps, the $\nu_*$ and $\rho_*$ expressions in \eqref{Eq:Lm:OptimalSystemIndices} can also be proved. 
\end{proof}

\paragraph{\textbf{Prior Work}}
The proposed optimization problems in \eqref{Eq:Lm:OptimalSystemIndices} are functional optimization problems (due to the program variable: $u_{t_0:t_1}\subset U$). Therefore, they are hard to solve without considerably simplifying the problem setup via introducing various limitations and assumptions. 

For example, the work in \cite{Koch2021}: (i) assumes the system dynamics to be linear, (ii) fixes $t_0=0$, $t_1=N$ (where $N$ is a predefined value) and $x_0=\0$, (iii) omits the effect of storage function $(S(x_1)-S(x_0))$, and only then develops an off-line gradient-based iterative process to update the control input profile $u_{t_0:t_1}$ so as to solve the corresponding (significantly simplified) version of \eqref{Eq:Lm:OptimalSystemIndices}. It is also worth noting that, apart from the said limitations/assumptions used in \cite{Koch2021}, it also assumes the ability to control/simulate the system dynamics under arbitrary control input profiles - which may not be possible in many real-world complex systems of interest.
 
In contrast to the off-line solution proposed in \cite{Koch2021}, the work in \cite{Zakeri2019} proposes an on-line solution (for \eqref{Eq:Lm:OptimalSystemIndices}). In particular, it: (i) assumes the system dynamics to be non-linear, (ii) considers the control input profile $u_{t_0:t_1}$ and the initial state $x_0$ as given, (iii) fixes $t_0=t_o$ and $t_1=t_f$ (where $t_o$ and $t_f$ are the initial and current time values, respectively), (iv) omits the effect of storage function $(S(x_1)-S(x_0))$, and then provides estimates for the optimal system indices $\gamma_*$, $\nu_*$ and $\rho_*$ \eqref{Eq:Lm:OptimalSystemIndices} respectively as $\hat{\gamma}$, $\hat{\nu}$ and $\hat{\rho}$, where 
\begin{equation}\label{Eq:EstimatedOptimalSystemIndicesOld}
    \begin{gathered}
        \hat{\gamma}^2 = \frac{\int_{t_o}^{t_f} y^\T(t)y(t)dt}{\int_{t_o}^{t_f}u^\T(t)u(t)dt}, \\
        \hat{\nu} = \frac{\int_{t_o}^{t_f} u^\T(t)y(t)dt}{\int_{t_o}^{t_f} u^\T(t)u(t)dt},\\
        \hat{\rho} = \frac{\int_{t_o}^{t_f} u^\T(t)y(t)dt}{\int_{t_o}^{t_f} y^\T(t)y(t)dt}.
    \end{gathered}
\end{equation}

Note that the optimal L2G estimate $\hat{\gamma}^2$ in \eqref{Eq:EstimatedOptimalSystemIndicesOld} is identical to the objective function value $\psi_\gamma(t_0,t_1,x_0,u_{t_0:t_1})$ in \eqref{Eq:Lm:OptimalSystemIndices} evaluated directly at $t_0=t_o$, $t_1=t_f$, $x_0=x(t_o)$, $u_{t_0:t_1} = u_{t_o:t_f}$ while assuming $(S(x_1)-S(x_0))=0$ (note also that, when input-output profiles of the system are known over $t\in[t_o,t_f]$, the state information $x(t), \forall t\in[t_o,t_f]$ is not explicitly required to evaluate $\hat{\gamma}$ in \eqref{Eq:EstimatedOptimalSystemIndicesOld}). The same arguments are valid for the optimal IFP and OFP estimates $\hat{\nu}$ and $\hat{\rho}$ in \eqref{Eq:EstimatedOptimalSystemIndicesOld}. 

In essence, $\hat{\gamma}$, $\hat{\nu}$ and $\hat{\rho}$ in \eqref{Eq:EstimatedOptimalSystemIndicesOld} can respectively be seen as the values of the objective functions $\psi_\gamma$, $\psi_\nu$ and $\psi_\rho$ in  \eqref{Eq:Lm:OptimalSystemIndices} at a particular point in $\Psi_\gamma \cap \Psi_\nu \cap \Psi_\rho$.  
Therefore, it is clear that the approach proposed in \cite{Zakeri2019} may suffer from inaccuracies due to: (i) the lack of any optimization stage (like in \eqref{Eq:Lm:OptimalSystemIndices} and \cite{Koch2021}), (ii) the omission of the storage function, and (iii) the susceptibility to the saturation effects (as $t_f\rightarrow \infty$). In the sequel, we propose an alternative approach to \cite{Zakeri2019} that addresses the above challenges when solving \eqref{Eq:Lm:OptimalSystemIndices} on-line.

\paragraph{\textbf{On-Line Estimates for Optimal System Indices}}

The proposed on-line solution for \eqref{Eq:Lm:OptimalSystemIndices} in this paper requires the following assumption. 
\begin{assumption}\label{As:LipschitzStorageFunction}
The rate of change of the storage function $S:X\rightarrow \R_{\geq 0}$ is bounded such that $\vert \frac{dS(x(t))}{dt}  \vert \leq K_s$ over the period $t\in [t_o,t_f)$, where $K_s$ is a known constant.
\end{assumption}

Similar to \cite{Zakeri2019}, we assume the system dynamics to be non-linear and consider the control input profile $u_{t_0:t_1}$ as a given. However, unlike in \cite{Zakeri2019}, we do not fix $t_0$ and $t_1$ values (to $t_o$ and $t_f$, respectively). Instead, we treat $t_0$ and $t_1$ as \emph{program variables} over which the we optimize the objective functions:
\begin{equation}\label{Eq:EstimatedFractionalFunctionsOfInterest}
\begin{aligned}
    \hat{\psi}_{\gamma}(t_0,t_1) \triangleq 
    \frac{\int_{t_0}^{t_1} y^\T(t)y(t)dt - K_s(t_1-t_0)}{\int_{t_0}^{t_1}u^\T(t)u(t)dt},\\
    \hat{\psi}_{\nu}(t_0,t_1) \triangleq 
    \frac{\int_{t_0}^{t_1} u^\T(t)y(t)dt + K_s(t_1-t_0)}{\int_{t_0}^{t_1}u^\T(t)u(t)dt},\\
    \hat{\psi}_{\rho}(t_0,t_1) \triangleq 
    \frac{\int_{t_0}^{t_1} u^\T(t)y(t)dt + K_s(t_1-t_0)}{\int_{t_0}^{t_1}y^\T(t)y(t)dt},\\
\end{aligned}
\end{equation}
on-line using past input-output information seen over a time period $[t_o,t_f) \supseteq [t_0,t_1)$. Note that $K_s$ in \eqref{Eq:EstimatedFractionalFunctionsOfInterest} is from As. \ref{As:LipschitzStorageFunction} (see also  \eqref{Eq:FractionalFunctionsOfInterest}). Moreover, in parallel to \eqref{Eq:FractionalFunctionsOfInterest}, the objective functions in \eqref{Eq:EstimatedFractionalFunctionsOfInterest} are respectively defined over the spaces 
\begin{equation*}\label{Eq:EstimatedFractionalFunctionsOfInterestSpaces}
    \begin{aligned}
\hat{\Psi}_\gamma \triangleq& \{(t_0,t_1):[t_0,t_1) \in [t_o,t_f), (t_1,t_0,x_0,u_{t_0:t_1})\in \Psi_\gamma\},\\ 
\hat{\Psi}_\nu \triangleq& \{(t_0,t_1):[t_0,t_1) \in [t_o,t_f), (t_1,t_0,x_0,u_{t_0:t_1})\in \Psi_\nu\},\\
\hat{\Psi}_\rho \triangleq& \{(t_0,t_1):[t_0,t_1) \in [t_o,t_f), (t_1,t_0,x_0,u_{t_0:t_1})\in \Psi_\rho\}.    
    \end{aligned}
\end{equation*}

The following lemma defines the proposed estimates of the optimal system indices and establishes their relationships with the corresponding optimal system indices given in \eqref{Eq:Lm:OptimalSystemIndices}.

\begin{lemma}\label{Lm:EstimatedOptimalSystemIndices}
Using a given input-output profile of the system $\mathcal{H}$ \eqref{Eq:NonlinearSystem}, the optimal system indices \eqref{Eq:Lm:OptimalSystemIndices}: L2G $\gamma_*$, IFP $\nu_*$ and OFP $\rho_*$ can respectively be estimated by  
\begin{equation}\label{Eq:Lm:EstimatedOptimalSystemIndices}
\begin{gathered}
    \hat{\gamma}^2_* = \max_{(t_0,t_1) \in \hat{\Psi}_\gamma}\ \hat{\psi}_{\gamma}(t_0,t_1),\\   
    \hat{\nu}_* = \min_{(t_0,t_1) \in \hat{\Psi}_\nu}\ \hat{\psi}_{\nu}(t_0,t_1),\\
    \hat{\rho}_* = \min_{(t_0,t_1)\in \hat{\Psi}_\rho}\ \hat{\psi}_{\rho}(t_0,t_1),
\end{gathered}
\end{equation}
so that $\hat{\gamma}^2_* \leq \gamma^2_*$, $\hat{\nu}_* \geq \nu_*$ and $\hat{\rho}_* \geq \rho_*$ holds true.
\end{lemma}

\begin{proof}
According to As. \ref{As:LipschitzStorageFunction}, 
$$
-K_s(t_1-t_0) \leq S(x_1) - S(x_0) \leq K_s(t_1-t_0).
$$
Hence, under a given input-output profile, the objective functions $\hat{\psi}_{\gamma}(t_1,t_0)$ in \eqref{Eq:Lm:EstimatedOptimalSystemIndices} and $\psi_{\gamma}(t_1,t_0,x_0,u_{t_0:t_1})$ in \eqref{Eq:Lm:OptimalSystemIndices} (defined in \eqref{Eq:EstimatedFractionalFunctionsOfInterest} and \eqref{Eq:FractionalFunctionsOfInterest}, respectively) satisfy: 
$\hat{\psi}_{\gamma}(t_1,t_0) \leq  \psi_{\gamma}(t_1,t_0,x_0,u_{t_0:t_1})$. Moreover, the feasible spaces of the optimization problems in \eqref{Eq:Lm:EstimatedOptimalSystemIndices} are subsets of those in \eqref{Eq:Lm:OptimalSystemIndices} (e.g., $\hat{\Psi}_\gamma \subset \Psi_\gamma$). Starting from the $\hat{\gamma}^2_*$ expression in \eqref{Eq:Lm:EstimatedOptimalSystemIndices} and then applying the above two facts, we get 
\begin{equation*}
    \begin{aligned}
    \hat{\gamma}^2_* = \max_{(t_0,t_1) \in \hat{\Psi}_\gamma}\ \hat{\psi}_{\gamma}(t_0,t_1) 
    \leq \max_{(t_0,t_1) \in \hat{\Psi}_\gamma}\ \psi_{\gamma}(t_1,t_0,x_0,u_{t_0:t_1})\\
    \leq \max_{(t_0,t_1,x_0,u_{t_0:t_1}) \in \Psi_\gamma}\ \psi_{\gamma}(t_1,t_0,x_0,u_{t_0:t_1}) 
    = \gamma_*^2.
    \end{aligned}
\end{equation*}
Thus, $\hat{\gamma}^2_*\leq \gamma_*^2$. Following the same steps, we can also prove that $\hat{\nu}_* \geq \nu_*$ and $\hat{\rho}_* \geq \rho_*$. This completes the proof.
\end{proof}

It is worth noting that the inequality relationships such as $\hat{\rho}_* \geq \rho_*$ proven in Lm. \ref{Lm:EstimatedOptimalSystemIndices} cannot be established between the estimates proposed in \eqref{Eq:EstimatedOptimalSystemIndicesOld} \cite{Zakeri2019} and the optimal values stated in \eqref{Eq:Lm:OptimalSystemIndices} (e.g., between $\hat{\rho}$ and $\rho_*$) due to the omission of the storage function in \eqref{Eq:EstimatedOptimalSystemIndicesOld} \cite{Zakeri2019}. However, we will revisit this point at the end of the next section. 

In the next section, we propose an efficient solution for the fractional function optimization problems of the form \eqref{Eq:Lm:EstimatedOptimalSystemIndices}.


\section{Fractional Function Optimization}\label{Sec:FractionalFunctionOptimization}

In this section, we study a class of fractional function optimization problems (FFOPs) that includes all three problem forms in  \eqref{Eq:Lm:EstimatedOptimalSystemIndices}. The goal here is to develop an efficient solution for this class of FFOPs so as to systematically solve \eqref{Eq:Lm:EstimatedOptimalSystemIndices}.

\paragraph{\textbf{Preliminaries}} We first establish a minor algebraic result that will be used later on. 

\begin{lemma}\label{Lm:SimpleAlgebraicResult}
For any $a,c \in \R$ and $b,d\in\R_{>0}$, 
\begin{equation}\label{Eq:Lm:SimpleAlgebraicResult}
    \frac{a}{b} \lesseqqgtr \frac{a+c}{b+d} \lesseqqgtr \frac{c}{d}.
\end{equation}
\end{lemma}
\begin{proof}
First, we prove $\frac{a}{b} < \frac{a+c}{b+d} \iff \frac{a+c}{b+d} < \frac{c}{d}$. Note that, since $b>0$ and $b+d>0$, we can write
\begin{equation}\label{Eq:Lm:SimpleAlgebraicResultStep1}
  \frac{a}{b} < \frac{a+c}{b+d} \iff ab + ad < ab + bc \iff ad < bc.  
\end{equation}
Now, adding a $cd$ term to the both sides of the last inequality and using the fact that $d>0$ and $b+d > 0$, we can obtain
\begin{equation}\label{Eq:Lm:SimpleAlgebraicResultStep2}
  ad < bc \iff ad + cd < bc + cd \iff \frac{a+c}{b+d} < \frac{c}{d}. 
\end{equation}
Combining \eqref{Eq:Lm:SimpleAlgebraicResultStep1} and \eqref{Eq:Lm:SimpleAlgebraicResultStep2}, we get $\frac{a}{b} < \frac{a+c}{b+d} \iff \frac{a+c}{b+d} < \frac{c}{d}$.

Now, note that these same steps can be used to prove 
$\frac{a}{b} = \frac{a+c}{b+d} \iff \frac{a+c}{b+d} = \frac{c}{d}$ and  
$\frac{a}{b} > \frac{a+c}{b+d} \iff \frac{a+c}{b+d} > \frac{c}{d}$. 
Therefore, $\frac{a}{b} \lesseqqgtr \frac{a+c}{b+d} \iff \frac{a+c}{b+d} \lesseqqgtr \frac{c}{d}$.
\end{proof}

\paragraph{\textbf{The Fractional Function}}

Let $t_o,t_f\in\R_{\geq 0}$ where $t_o<t_f$ be two continuous-time instants and $p(t)$,$q(t)$ where $p:[t_o,t_f) \rightarrow \R$, $q:[t_o,t_f)\rightarrow \R_{\geq 0}$ be two continuous-time signals. Moreover, let $r(\tau)$, $s(\tau)$ where $r:\R_{>0}\rightarrow \R$, $s:\R_{>0}\rightarrow \R_{\geq 0}$ be two continuous linear mappings and $\dot{r}(0) \triangleq \lim_{\tau \rightarrow 0} \frac{dr(\tau)}{d\tau}$, $\dot{s}(0) \triangleq \lim_{\tau \rightarrow 0} \frac{ds(\tau)}{d\tau}$ where $\dot{r}(0)\in\R$, $\dot{s}(0)\in\R_{\geq 0}$ be two known constants.

Now, for any continuous-time interval $[t_0,t_1) \subseteq [t_o,t_f)$, let us define a corresponding \emph{fractional function} of the form:
\begin{equation}\label{Eq:GenericRationalFunctionCont}
    \psi(t_0,t_1) \triangleq
    \frac{\int_{t_0}^{t_1} p(t)dt + r(t_1-t_0)}{\int_{t_0}^{t_1} q(t)dt + s(t_1-t_0)}.
\end{equation}
To ensure that \eqref{Eq:GenericRationalFunctionCont} is well-defined over any $[t_0,t_1) \subseteq [t_o,t_f)$, we require the following assumption. 

\begin{assumption}\label{As:NoDeadZoneCont} (No-Dead-Zones) 
The continuous-time signal $q(t)$ involved in the fractional function \eqref{Eq:GenericRationalFunctionCont} is such that $q:[t_o,t_f)\rightarrow \R_{> 0}$, i.e., $q(t) >0, \forall t\in[t_o,t_f)$. 
\end{assumption}

Due to obvious reasons, we call this assumptions as the ``no-dead-zone'' (in $q(t)$) assumption. However, we also point out that we will relax this assumption later on in this section.

Under As. \ref{As:NoDeadZoneCont}, note that the denominator of the fractional function $\psi(t_0,t_1)$ defined in \eqref{Eq:GenericRationalFunctionCont} will always be positive definite. Next, we establish the following lemma. 

\begin{lemma}\label{Lm:RationalFunctionPropertyCont}
For any fractional function $\psi(t_0,t_1)$ of the form \eqref{Eq:GenericRationalFunctionCont} under As. \ref{As:NoDeadZoneCont}: 
\begin{equation}\label{Eq:Lm:RationalFunctionPropertyCont}
    \psi(t_0,t_m) \lesseqqgtr \psi(t_0,t_1) \lesseqqgtr \psi(t_m,t_1),
\end{equation}
for any $t_m\in\R_{>0}$ such that $t_0<t_m<t_1$. 
\end{lemma}

\begin{proof}
Since $t_0<t_1$ (by definition), there always exists $t_m\in\R_{>0}$ such that $t_0 < t_m < t_1$. Using this fact and the nature of the fractional function $\psi(\cdot,\cdot)$ defined in \eqref{Eq:GenericRationalFunctionCont}, we can re-state the terms $\psi(t_0,t_m)$, $\psi(t_0,t_1)$ and $\psi(t_m,t_1)$ as 
\begin{equation}\label{Eq:Lm:RationalFunctionPropertyContStep0}
\psi(t_0,t_m) = \frac{a}{b},\ \ 
\psi(t_0,t_1) = \frac{a+c}{b+d}\ \ \mbox{ and }\ \ 
\psi(t_m,t_1) = \frac{c}{d},
\end{equation}
respectively, where 
$a \triangleq \int_{t_0}^{t_m}p(t)dt + r(t_m-t_0)$,
$b \triangleq \int_{t_0}^{t_m}q(t)dt + s(t_m-t_0)$,
$c \triangleq \int_{t_m}^{t_1}p(t)dt + r(t_1-t_m)$ and 
$d \triangleq \int_{t_m}^{t_1}q(t)dt + s(t_1-t_m)$.    
Note that, $a,c\in\R$, and under As. \ref{As:NoDeadZoneCont}, $b,d\in\R_{>0}$. Therefore, Lm. \ref{Lm:SimpleAlgebraicResult} can now be applied to compare the terms involved in \eqref{Eq:Lm:RationalFunctionPropertyContStep0} as 
\begin{equation*}
    \frac{a}{b} \lesseqqgtr \frac{a+c}{b+d} \lesseqqgtr \frac{c}{d} 
    \iff
    \psi(t_0,t_m) \lesseqqgtr \psi(t_0,t_1) \lesseqqgtr \psi(t_m,t_1),
\end{equation*}
which completes the proof.
\end{proof}

\paragraph{\textbf{Fractional Function Optimization}}
Under As. \ref{As:NoDeadZoneCont}, let us denote the space over which a fractional function $\psi(t_0,t_1)$ of the form \eqref{Eq:GenericRationalFunctionCont} is well-defined as  
\begin{equation}\label{Eq:GenericRationalFunctionContSpace}
    \Psi \triangleq \{(t_0,t_1):[t_0,t_1) \subseteq [t_o,t_f)\}.
\end{equation}
In the following theorem, we establish our main theoretical result on optimizing a fractional function $\psi(t_0,t_1)$ of the form \eqref{Eq:GenericRationalFunctionCont} over all $(t_0,t_1)\in\Psi$.

\begin{theorem}\label{Th:RationalFunctionPropertyCont}
Any fractional function $\psi(t_0,t_1)$ of the form \eqref{Eq:GenericRationalFunctionCont} under As. \ref{As:NoDeadZoneCont} can be efficiently optimized over the corresponding space $(t_0,t_1)\in\Psi$ exploiting the relationships:
\begin{equation}\label{Eq:Th:RationalFunctionPropertyMaxCont}
    \Big\{ 
    \underset{(t_0,t_1)\in\Psi}{\max}\ 
    \psi(t_0,t_1)  
    \Big\} 
    \equiv
    \Big\{
    \underset{t\in [t_o,t_f)}{\max}\ 
    \lim_{\Delta \rightarrow 0} \psi(t,t+\Delta) 
    \Big\},
\end{equation}
\begin{equation}\label{Eq:Th:RationalFunctionPropertyMinCont}
    \Big\{ 
    \underset{(t_0,t_1)\in\Psi}{\min}\ 
    \psi(t_0,t_1) 
    \Big\} 
    \equiv
    \Big\{
    \underset{t\in [t_o,t_f)}{\min}\ 
    \lim_{\Delta \rightarrow 0} \psi(t,t+\Delta) 
    \Big\},
\end{equation}
where 
\begin{equation*}
    \lim_{\Delta \rightarrow 0} \psi(t,t+\Delta) = \frac{p(t)+\dot{r}(0)}{q(t)+\dot{s}(0)}.
\end{equation*}
\end{theorem}

\begin{proof}
We will first prove the equivalence relation in \eqref{Eq:Th:RationalFunctionPropertyMaxCont}. Let $t_0^*,t_1^*$ be the optimal arguments corresponding to the optimization problem in the left hand side (LHS) of \eqref{Eq:Th:RationalFunctionPropertyMaxCont}. Note that, irrespective of $t_0$,  
\begin{align*}
    \lim_{t_1 \rightarrow t_0} \psi(t_0,t_1) 
    = &\lim_{t_1 \rightarrow t_0} \frac{\int_{t_0}^{t_1}p(t)dt+r(t_1-t_0)}{\int_{t_0}^{t_1}q(t)dt+s(t_1-t_0)}\\
    = & \lim_{t_1 \rightarrow t_0} \frac{\frac{d}{dt_1}\int_{t_0}^{t_1}p(t)dt+ \frac{d}{dt_1}r(t_1-t_0)}{\frac{d}{dt_1}\int_{t_0}^{t_1}q(t)dt+\frac{d}{dt_1}s(t_1-t_0)}\\
    = & \lim_{t_1 \rightarrow t_0} \frac{p(t_1)+ \dot{r}(t_1-t_0)}{q(t_1)+\dot{s}(t_1-t_0)}
    = \frac{p(t_0)+ \dot{r}(0)}{q(t_0)+\dot{s}(0)}.
\end{align*}
In the above simplification, we have used the L'Hospital's rule as $r(0)=0$, $s(0)=0$ (recall that $r(\tau),s(\tau)$ are linear mappings). 
From this result, it is clear that the equivalence relation in \eqref{Eq:Th:RationalFunctionPropertyMaxCont} is implied if $t_0^*,t_1^*$ are such that $t_1^*=t_0^*+\epsilon$ where $\epsilon$ is an infinitesimally small positive number. In other words, \eqref{Eq:Th:RationalFunctionPropertyMaxCont} is proved if we can prove that $t_1^*$ is infinitesimally close to $t_0^*$. In the sequel, we prove this using a contradiction. 

Let us assume $t_0^*$ and $t_1^*$ are not infinitesimally close. Then, there should exist $t_m\in\R_{>0}$ such that $t_0^*<t_m<t_1^*$. According to Lm. \ref{Lm:RationalFunctionPropertyCont}, for any such $t_m$, either: 
\begin{equation}
    \begin{cases}
    \psi(t_0^*,t_1^*) < \psi(t_m,t_1^*),\\
    \psi(t_0^*,t_1^*) = \psi(t_0^*,t_m) = \psi(t_m,t_1^*),\ \ \mbox{ or }\ \ \\
    \psi(t_0^*,t_1^*) < \psi(t_0^*,t_m).
    \end{cases}
\end{equation}
In other words, $\psi(t_0^*,t_1^*) \leq \min\{\psi(t_0^*,t_m),  \psi(t_m,t_1^*)\}$. This implies that $\psi(t_0^*,t_1^*)$ is sub-optimal - which is a contradiction. Therefore, it is clear that the optimal arguments $t_0^*,t_1^*$ should be infinitesimally close. As stated earlier, this completes the proof of \eqref{Eq:Th:RationalFunctionPropertyMaxCont}. The same steps can be used to prove the equivalence condition in  \eqref{Eq:Th:RationalFunctionPropertyMinCont}.
\end{proof}

It is worth noting that the application of Th. \ref{Th:RationalFunctionPropertyCont} can significantly reduce the computational cost associated with optimizing fractional functions of the form \eqref{Eq:GenericRationalFunctionCont}. In particular, it can reduce the search space by one dimension (see \eqref{Eq:Th:RationalFunctionPropertyMaxCont} and \eqref{Eq:Th:RationalFunctionPropertyMinCont}). Note also that the resulting simplified optimization problems take a particular form that can easily be implemented in an on-line setting without using much memory.

However, note that, if we relaxed As. \ref{As:NoDeadZoneCont}, we no longer can consider optimizing fractional functions of the form \eqref{Eq:GenericRationalFunctionCont} over the space $\Psi$ \eqref{Eq:GenericRationalFunctionContSpace}. Instead, we will have to optimize such functions over a confined space $\Psi\backslash\Psi_q$ (where $\Psi_q$ will be defined in the sequel). Hence, Th. \ref{Th:RationalFunctionPropertyCont} needs to be generalized to address situations where As. \ref{As:NoDeadZoneCont} does not hold.

\paragraph{\textbf{Relaxing As. \ref{As:NoDeadZoneCont}}}

When As. \ref{As:NoDeadZoneCont} is relaxed, the fractional function $\psi(t_0,t_1)$ in \eqref{Eq:GenericRationalFunctionCont} becomes ill-defined over intervals $[t_0,t_1)\subseteq [t_o,t_f)$ where $\int_{t_0}^{t_1}q(t)dt + s(t_1-t_0) = 0$. Recall that, by definition, $q:[t_o,t_f)\rightarrow \R_{\geq 0}$ is an arbitrary signal and $s:\R_{>0}\rightarrow \R_{\geq 0}$ is a linear mapping. Therefore, $\int_{t_0}^{t_1}q(t)dt + s(t_1-t_0) = 0 \iff q(t) = 0, \forall t\in[t_0,t_1)$ and $s(\tau)=0, \forall \tau \in \R_{>0}$. This inspires the following assumption.

\begin{assumption}\label{As:NoLinearMappingInDenominator}
The linear mapping $s(\tau)$ involved in the fractional function \eqref{Eq:GenericRationalFunctionCont} is such that $s(\tau)=0, \forall \tau \in \R_{>0}$.
\end{assumption}

Recalling $\Psi$ \eqref{Eq:GenericRationalFunctionContSpace}, we define 
$$
\Psi_q \triangleq \{(t_0,t_1)\in\Psi: q(t)=0, \forall t\in[t_0,t_1)\}
$$
as the space of ``sub-dead-zones'' (i.e., intervals $[t_0,t_1) \subseteq [t_o,t_f)$ where $q(t) = 0, \forall t\in [t_0,t_1)$).  Note that any subset of a sub-dead-zone will also be a sub-dead-zone, and thus, $\Psi_q$ may contain infinitely many $(t_0,t_1)$ tuples. We also define 
$$\bar{\Psi}_q \triangleq \{(t_0,t_1)\in\Psi_q: q(t_0-\epsilon)>0, q(t_1)>0\},$$
where $\epsilon$ is an infinitesimally small constant, as the space of ``dead-zones'' (i.e., the largest possible sub-dead-zones). Consequently, $\bar{\Psi}_q$ only contains a finite number of $(t_0,t_1)$ tuples and $\bar{\Psi}_q$ is empty if there are no dead-zones inside $[t_o,t_f)$. For notational convenience, let us also define 
$$\tilde{\Psi}_q \triangleq \{t: t \in [t_o,t_f), q(t)=0\}$$ 
as the set of time instants corresponding to the dead-zones (note that $\tilde{\Psi}_q = \cup_{(t_0,t_1)\in\bar{\Psi}_q}[t_0,t_1) = \cup_{(t_0,t_1)\in \Psi_q}[t_0,t_1)$). 

The following theorem generalizes Th. \ref{Th:RationalFunctionPropertyCont} to address the presence of dead-zones inside $[t_o,t_f)$ (i.e., when $\vert \bar{\Psi}_q \vert >0$).



\begin{table*}
\hrule 
\begin{align}\label{Eq:Co:RationalFunctionPropertyMaxCont}
    &\Big\{ 
    \underset{(t_0,t_1) \in \Psi \backslash \Psi_q }{\max}\ 
    \psi(t_0,t_1)  
    \Big\} 
    \equiv\ 
    \max 
    \Big\{
    \big\{
    \underset{t \in [t_o,t_f)\backslash \tilde{\Psi}_q}{\max}\ 
    \lim_{\Delta \rightarrow 0} \psi(t,t+\Delta) 
    \big\},\ 
    \big\{
    \underset{\substack{t\in [t_{zo},t_{zf})\\(t_{zo},t_{zf})\in\bar{\Psi}_q}}{\max}\ 
    \lim_{\Delta \rightarrow 0} \psi(t_{zo}-\Delta,t)
    \big\},\ 
    \big\{
    \underset{\substack{t\in [t_{zo},t_{zf})\\(t_{zo},t_{zf})\in\bar{\Psi}_q}}{\max}\ 
    \lim_{\Delta \rightarrow 0} \psi(t,t_{zf}+\Delta) 
    \big\}
    \Big\},
    \\
    \label{Eq:Co:RationalFunctionPropertyMinCont}
    &\Big\{ 
    \underset{(t_0,t_1)\in\Psi\backslash \Psi_q}{\min}\ 
    \psi(t_0,t_1)  
    \Big\} 
    \equiv\ 
    \min 
    \Big\{
    \big\{
    \underset{t \in [t_o,t_f)\backslash \tilde{\Psi}_q}{\min}\ 
    \lim_{\Delta \rightarrow 0} \psi(t,t+\Delta) 
    \big\},\ 
    \big\{
    \underset{\substack{t\in [t_{zo},t_{zf})\\(t_{zo},t_{zf})\in\bar{\Psi}_q}}{\min}\ 
    \lim_{\Delta \rightarrow 0} \psi(t_{zo}-\Delta,t)
    \big\},\  
    \big\{
    \underset{\substack{t\in [t_{zo},t_{zf})\\(t_{zo},t_{zf})\in\bar{\Psi}_q}}{\min}\ 
    \lim_{\Delta \rightarrow 0} \psi(t,t_{zf}+\Delta)
    \big\}
    \Big\},\\ \label{Eq:Co:RationalFunctionPropertyObjectives1}
    &\psi(t_0,t_1) 
    \triangleq\  \frac{\int_{t_0}^{t_1}p(t)dt+r(t_1-t_0)}{\int_{t_0}^{t_1}q(t)dt}, 
    \ \ \ \mbox{(i.e., \eqref{Eq:GenericRationalFunctionCont} under As. \ref{As:NoLinearMappingInDenominator}) \ \ \ and }\ \ \ 
    \lim_{\Delta \rightarrow 0} \psi(t,t+\Delta) 
    = \frac{p(t)+\dot{r}(0)}{q(t)},\\ \label{Eq:Co:RationalFunctionPropertyObjectives2}
    &\lim_{\Delta \rightarrow 0} \psi(t_{zo}-\Delta,t) 
    =   
    \begin{cases}
    +\infty, \mbox{ if } \int_{t_{zo}}^t p(\tau)d\tau + r(t-t_{zo}) > 0,\\
    -\infty, \mbox{ if } \int_{t_{zo}}^t p(\tau)d\tau + r(t-t_{zo}) < 0,
    \end{cases}
    \mbox{ and }\ \ \ 
    \lim_{\Delta \rightarrow 0} \psi(t,t_{zf}+\Delta) 
    = 
    \begin{cases}
    +\infty, \mbox{ if } \int_t^{t_{zf}} p(\tau)d\tau + r(t_{zf}-t) > 0,\\
    -\infty, \mbox{ if } \int_t^{t_{zf}} p(\tau)d\tau + r(t_{zf}-t) < 0.
    \end{cases}
\end{align}
\hrule
\end{table*}

\begin{theorem}\label{Th:Th:RationalFunctionPropertyCont}
Any fractional function $\psi(t_0,t_1)$ of the form \eqref{Eq:GenericRationalFunctionCont} under As. \ref{As:NoLinearMappingInDenominator} can be efficiently optimized over the corresponding space $(t_0,t_1)\in\Psi\backslash \Psi_q$ exploiting the relationships: \eqref{Eq:Co:RationalFunctionPropertyMaxCont} and \eqref{Eq:Co:RationalFunctionPropertyMinCont} (under the objective functions in  \eqref{Eq:Co:RationalFunctionPropertyObjectives1} and \eqref{Eq:Co:RationalFunctionPropertyObjectives2}). 
\end{theorem}
\begin{proof}
The main idea behind this proof is to consider all possible intervals inside $[t_o,t_f)$ that are: (i) not sub-dead-zones (i.e., have well-defined $\psi(\cdot,\cdot)$ values) and (ii) not further breakable into two or more non-sub-dead-zones (as they would be sub-optimal according to Lm. \ref{Lm:RationalFunctionPropertyCont}). It is easy to see that there are three types of such intervals:
\begin{enumerate}
    \item $\lbrack t,t+\epsilon),\ \ \forall t \in [t_o,t_f) \backslash \tilde{\Psi}_q,$
    \item $\lbrack t_{zo}-\epsilon,t),\ \ \forall t \in [t_{zo},t_{zf}), \forall (t_{zo},t_{zf})\in\bar{\Psi}_q,$
    \item $\lbrack t,t_{zf}+\epsilon),\ \ \forall t \in [t_{zo},t_{zf}), \forall (t_{zo},t_{zf})\in\bar{\Psi}_q,$
\end{enumerate}
where $\epsilon$ is an infinitesimally small positive number. These three interval types give the three different inner optimization problems in the right hand side (RHS) of \eqref{Eq:Co:RationalFunctionPropertyMaxCont} and \eqref{Eq:Co:RationalFunctionPropertyMinCont}. Finally, the simplified objective function expressions in \eqref{Eq:Co:RationalFunctionPropertyObjectives1} and \eqref{Eq:Co:RationalFunctionPropertyObjectives2} can be obtained using the L'Hospital's rule following the same steps as in the proof of Th. \ref{Th:RationalFunctionPropertyCont}.  
\end{proof}

From comparing \eqref{Eq:Co:RationalFunctionPropertyMaxCont}, \eqref{Eq:Co:RationalFunctionPropertyMinCont} in Th. \ref{Th:Th:RationalFunctionPropertyCont} with \eqref{Eq:Th:RationalFunctionPropertyMaxCont}, \eqref{Eq:Th:RationalFunctionPropertyMinCont} in Th. \ref{Th:RationalFunctionPropertyCont}, it should be clear that the main consequence of the presence of dead-zones is the last two inner optimization problems in the RHS of \eqref{Eq:Co:RationalFunctionPropertyMaxCont}, \eqref{Eq:Co:RationalFunctionPropertyMinCont}. Note that, based on \eqref{Eq:Co:RationalFunctionPropertyObjectives2}, these additional inner optimization problems can be either decisive or irrelevant (but never in-between).

The following corollary shows that, under some special conditions, when optimizing a fractional function of the form \eqref{Eq:GenericRationalFunctionCont}, the dead-zones can be conveniently disregarded. 

\begin{assumption}\label{As:NumeratorDenominatorConnection}
The signals $p(t)$ and $q(t)$ involved in the fractional function \eqref{Eq:GenericRationalFunctionCont} are such that $q(t)=0 \implies p(t) = 0$ for any $t\in[t_o,t_f)$. 
\end{assumption}

\begin{corollary}\label{Co:Th:RationalFunctionPropertyCont2}
Any fractional function $\psi(t_0,t_1)$ of the form \eqref{Eq:GenericRationalFunctionCont} under As. \ref{As:NoLinearMappingInDenominator} and As. \ref{As:NumeratorDenominatorConnection} can be efficiently optimized over the corresponding space $(t_0,t_1)\in\Psi\backslash\Psi_q$ exploiting the relationships: 
(i) If $r(\tau) \leq 0, \forall \tau \in\R_{>0}$: 
\begin{equation}\label{Eq:Co:Th:RationalFunctionPropertyCont2Max}
    \Big\{ 
    \underset{(t_0,t_1) \in \Psi\backslash\Psi_q}{\max}\ 
    \psi(t_0,t_1) 
    \Big\} 
    \equiv
    \Big\{
    \underset{t \in [t_o,t_f)\backslash \tilde{\Psi}_q}{\max}\ 
    \lim_{\Delta \rightarrow 0} \psi(t,t+\Delta) 
    \Big\},
\end{equation}
(ii) If $r(\tau) \geq 0, \forall \tau \in\R_{>0}$:
\begin{equation}\label{Eq:Co:Th:RationalFunctionPropertyCont2Min}
    \Big\{ 
    \underset{(t_0,t_1) \in \Psi\backslash\Psi_q}{\min}\ 
    \psi(t_0,t_1) 
    \Big\} 
    \equiv
    \Big\{
    \underset{t \in [t_o,t_f)\backslash \tilde{\Psi}_q}{\min}\ 
    \lim_{\Delta \rightarrow 0} \psi(t,t+\Delta) 
    \Big\},
\end{equation}
where 
\begin{equation*}
    \lim_{\Delta \rightarrow 0} \psi(t,t+\Delta) = \frac{p(t)+\dot{r}(0)}{q(t)}.
\end{equation*}
\end{corollary}
\begin{proof}
These results directly follow from applying the assumed conditions in As. \ref{As:NumeratorDenominatorConnection} and Co. \ref{Co:Th:RationalFunctionPropertyCont2} in Th. \ref{Th:Th:RationalFunctionPropertyCont}. 
\end{proof}

We are now ready to use the developed theories for FFOPs to solve the optimization problems formulated in \eqref{Eq:Lm:EstimatedOptimalSystemIndices}.

\paragraph{\textbf{Application to Solve \eqref{Eq:Lm:EstimatedOptimalSystemIndices}}}
First, following the same notation $\tilde{\Psi}_q$ used before (e.g., see \eqref{Eq:Co:Th:RationalFunctionPropertyCont2Max}), let us define 
$$\tilde{\Psi}_u \triangleq \{t: t\in [t_o,t_f), u(t)=0\},$$ 
$$\tilde{\Psi}_y \triangleq \{t: t\in [t_o,t_f), y(t)=0\}.$$
The following theorem provides an efficient approach to solve the optimization problems formulated in \eqref{Eq:Lm:EstimatedOptimalSystemIndices}.

\begin{theorem}\label{Th:EstimatedOptimalSystemIndices}
Using a given input-output profile of the system $\mathcal{H}$ \eqref{Eq:NonlinearSystem}, the optimal system indices \eqref{Eq:Lm:OptimalSystemIndices}: L2G $\gamma_*$, IFP $\nu_*$ and OFP $\rho_*$ can respectively be estimated by  
\begin{equation}\label{Eq:Th:EstimatedOptimalSystemIndices}
\begin{gathered}
    \hat{\gamma}^2_* = \max_{t \in [t_o,t_f)\backslash \tilde{\Psi}_u} 
    \frac{y^\T(t)y(t) - K_s}{u^\T(t)u(t)},\\   
    \hat{\nu}_* = \min_{t \in [t_o,t_f)\backslash \tilde{\Psi}_u} 
    \frac{u^\T(t)y(t) + K_s}{u^\T(t)u(t)},\\
    \hat{\rho}_* = \min_{t \in [t_o,t_f)\backslash \tilde{\Psi}_y}
    \frac{u^\T(t)y(t) + K_s}{y^\T(t)y(t)},
\end{gathered}
\end{equation}
so that $\hat{\gamma}^2_* \leq \gamma^2_*$, $\hat{\nu}_* \geq \nu_*$ and $\hat{\rho}_* \geq \rho_*$ holds true.
\end{theorem}
\begin{proof}
Due to Lm. \ref{Lm:EstimatedOptimalSystemIndices}, here we only need to prove the equivalence of the respective optimization problems in \eqref{Eq:Lm:EstimatedOptimalSystemIndices} and \eqref{Eq:Th:EstimatedOptimalSystemIndices}. To this end, we use the proposed Th. \ref{Th:Th:RationalFunctionPropertyCont} and Co. \ref{Co:Th:RationalFunctionPropertyCont2}. 

First, note that As. \ref{As:NoLinearMappingInDenominator} holds for any fractional function defined in \eqref{Eq:EstimatedFractionalFunctionsOfInterest} and As. \ref{As:NumeratorDenominatorConnection} holds for the fractional functions $\hat{\psi}_\nu$ and $\hat{\psi}_\rho$ defined in \eqref{Eq:EstimatedFractionalFunctionsOfInterest} (as $u(t)=\0$ or $y(t)=\0$ implies $u^\T(t)y(t)=\0$). Hence, Co. \ref{Co:Th:RationalFunctionPropertyCont2} (in particular, the equivalence relationship in \eqref{Eq:Co:Th:RationalFunctionPropertyCont2Min}) can be used to efficiently evaluate the estimates for the optimal system indices: IFP $\hat{\nu}_*$ and OFP $\hat{\rho}_*$ given in \eqref{Eq:Lm:EstimatedOptimalSystemIndices}. 

Therefore, by applying $p(t)=u^\T(t)y(t)$, $q(t)=u^\T(t)u(t)$ and $r(\tau) = K_s(\tau)$ in \eqref{Eq:Co:Th:RationalFunctionPropertyCont2Min}, we can obtain the required equivalence of $\hat{\nu}_*$ expressions given in \eqref{Eq:Lm:EstimatedOptimalSystemIndices} and \eqref{Eq:Th:EstimatedOptimalSystemIndices}. Similarly, by applying $p(t)=u^\T(t)y(t)$, $q(t)=y^\T(t)y(t)$ and $r(\tau) = K_s(\tau)$ in \eqref{Eq:Co:Th:RationalFunctionPropertyCont2Min}, we can obtain the required equivalence of $\hat{\rho}_*$ expressions given in \eqref{Eq:Lm:EstimatedOptimalSystemIndices} and \eqref{Eq:Th:EstimatedOptimalSystemIndices}. 

On the other hand, note that, As. \ref{As:NumeratorDenominatorConnection} does not necessarily hold for the fractional function $\hat{\psi}_\gamma$ defined in \eqref{Eq:EstimatedFractionalFunctionsOfInterest} (as $u(t)=\0$ does not imply $y^\T(t)y(t)=\0$). However, since As. \ref{As:NoLinearMappingInDenominator} holds for this fractional function $\hat{\psi}_\gamma$, we can apply Th. \ref{Th:Th:RationalFunctionPropertyCont} (in particular, the equivalence relationship in \eqref{Eq:Co:RationalFunctionPropertyMaxCont}) to efficiently evaluate an estimate for the optimal system index: L2G $\hat{\gamma}_*$ given in \eqref{Eq:Lm:EstimatedOptimalSystemIndices}. Nevertheless, if we assumed that the system $\mathcal{H}$ \eqref{Eq:NonlinearSystem} has a finite L2G value, then, by applying \eqref{Eq:Co:RationalFunctionPropertyObjectives2} in \eqref{Eq:Co:RationalFunctionPropertyMaxCont}, we can obtain the same equivalence relationship \eqref{Eq:Co:Th:RationalFunctionPropertyCont2Max} proposed in Co. \ref{Co:Th:RationalFunctionPropertyCont2}. Hence, Co. \ref{Co:Th:RationalFunctionPropertyCont2} (in particular \eqref{Eq:Co:Th:RationalFunctionPropertyCont2Max}) can also be used here to efficiently evaluate $\hat{\gamma}_*$ given in \eqref{Eq:Lm:EstimatedOptimalSystemIndices}. 

Therefore, by applying $p(t)=y^\T(t)y(t)$, $q(t)=u^\T(t)u(t)$ and $r(\tau) = -K_s(\tau)$ in \eqref{Eq:Co:Th:RationalFunctionPropertyCont2Max}, we can obtain the required equivalence of $\hat{\gamma}_*$ expressions given in \eqref{Eq:Lm:EstimatedOptimalSystemIndices} and \eqref{Eq:Th:EstimatedOptimalSystemIndices}. This completes the proof.
\end{proof}

\paragraph{\textbf{Discussion}} 
Comparing Lm. \ref{Lm:EstimatedOptimalSystemIndices} \eqref{Eq:Lm:EstimatedOptimalSystemIndices} and Th. \ref{Th:EstimatedOptimalSystemIndices} \eqref{Eq:Th:EstimatedOptimalSystemIndices}, we can make the following set of remarks.

\begin{remark}\textbf{(Efficiency)} \label{Rm:Efficiency}
Compared to \eqref{Eq:Lm:EstimatedOptimalSystemIndices}, \eqref{Eq:Th:EstimatedOptimalSystemIndices} is significantly computationally inefficient due to: (i) the reduced search space size, (ii) not having to compute/analyze continuous integrals, and (iii) the ability to implement on-line using less amount of memory. 
\end{remark}

\begin{remark}\textbf{(Sampling)}\label{Rm:Sampling}
When using a digital processor to compute the estimates of the optimal system indices (via \eqref{Eq:Lm:EstimatedOptimalSystemIndices} or \eqref{Eq:Th:EstimatedOptimalSystemIndices}), we cannot use the complete continuous input-output signals $\{u(t):t\in [t_o,t_f)\}$ and $\{y(t): t\in [t_o,t_f)\}$. Instead, we may only use sampled sets of data points from those signals $\{u(t): t \in D_s \subset [t_o,t_f)\}$ and $\{y(t): t \in D_s\subset [t_0,t_f)\}$, where $D_s$ is a \textbf{discrete} (and possibly \textbf{non-uniform}) set of sampling time instants selected from the continuous-time interval $[t_o,t_f)$. In a such scenario, note that we still can use \eqref{Eq:Th:EstimatedOptimalSystemIndices} (as opposed to \eqref{Eq:Lm:EstimatedOptimalSystemIndices}) and the bounds stated in Th. \ref{Th:EstimatedOptimalSystemIndices} will hold true. This is because, sampling will only reduce the feasible space size of the optimization problems in \eqref{Eq:Th:EstimatedOptimalSystemIndices} leading to increase the error between the estimates and the optimal values. Note also that this estimation error will decrease as the number of samples (i.e., $\vert D_s \vert$) increases. 
\end{remark}

\begin{remark}\textbf{(Control for Estimation)}\label{Rm:Control}
Consider a scenario where signals $\{p(t)\in\R:t\in[t_o,t_f)\}$ and $\{q(t)\in\R_{\geq 0}:t\in[t_o,t_f)\}$ in the fractional function $\psi(\cdot, \cdot)$ \eqref{Eq:GenericRationalFunctionCont} are \emph{controllable} over time $t$ from their initial values $p(t_o)$ and $q(t_o)$ via selecting $\dot{p}(t) \triangleq \frac{dp(t)}{dt}$ and $\dot{q}(t) \triangleq \frac{dq(t)}{dt}$ values, respectively. Note that, under As. \ref{As:NoDeadZoneCont} and As. \ref{As:NoLinearMappingInDenominator}, the RHS of \eqref{Eq:Th:RationalFunctionPropertyMaxCont} (equally applies for \eqref{Eq:Th:RationalFunctionPropertyMinCont}) can be written as  
\begin{equation}\label{Eq:ControlForEstimation1}
    \max_{t\in[t_o,t_f)} \psi(t) \triangleq \frac{p(t)+\dot{r}(0)}{q(t)}. 
\end{equation}
The dynamics of $\psi(t)$ defined above can be derived as
\begin{equation}
    \dot{\psi}(t) = -\frac{1}{q(t)}\left(\psi(t)\dot{q}(t)-\dot{p}(t)\right).
\end{equation}
Now, taking $x = [x_1,x_2,x_3]^\T \triangleq [\psi,p,q]^\T$ as the state variable and $u = [u_1,u_2]^\T \triangleq [\dot{p},\dot{q}]^\T$ as the control input, a non-linear state space representation $\dot{x} = g(x)u$, i.e.,
\begin{equation}
    \begin{bmatrix}
    \dot{x}_1 \\ \dot{x}_2 \\ \dot{x}_3
    \end{bmatrix}
    =
    \begin{bmatrix}
    \frac{1}{x_3} & -\frac{x_1}{x_3}\\
    1 & 0\\
    0 & 1
    \end{bmatrix}
    \begin{bmatrix}
    u_1 \\ u_2
    \end{bmatrix}, 
\end{equation}
can be obtained to represent the underlying dynamics. Using this setup, an optimal control problem can be formulated as 
\begin{equation}
\begin{aligned}
    \underset{{t_f,u_{t_o:t_f}}}{\arg\min}&\ \int_{t_o}^{t_f} L(x(t),u(t),t) dt + \phi(x(t_f),t_f),\\
    \mbox{such that }&\ \dot{x}(t) = g(x(t))u(t), \ \ \forall t\in [t_o,t_f)
\end{aligned}
\end{equation}
(e.g., with $L = u^\T(t) u(t)$ and $\phi = -x_1(t_f) + t_f = -\psi(t_f) + t_f$) to determine the optimal control input profile that rapidly leads to maximize $\psi(t)$ in \eqref{Eq:ControlForEstimation1}.
\end{remark}

\begin{remark}
\textbf{(Realizations)} \label{Rm:Realizations}
In reality, we might not have a single long input-output profile of the system to carry out the on-line estimation process (using \eqref{Eq:Lm:EstimatedOptimalSystemIndices} or \eqref{Eq:Th:EstimatedOptimalSystemIndices}). Instead, we might only have a collection of past realizations of the system (with arbitrary initial/terminal conditions) apart from a short recent input-output profile. In this kind of a setting, it is not clear how \eqref{Eq:Lm:EstimatedOptimalSystemIndices} can be used across all the realizations. However, in such a setting, \eqref{Eq:Th:EstimatedOptimalSystemIndices} can easily be used by treating all the realizations as a single long input-output profile (sequenced in an arbitrary order).
\end{remark}

\paragraph{\textbf{Selection of $K_s$}} 
For the system indices, as opposed to the proposed estimates in \eqref{Eq:EstimatedOptimalSystemIndicesOld} \cite{Zakeri2019}, the proposed estimates in this paper \eqref{Eq:Th:EstimatedOptimalSystemIndices} are dependent on the rate bound $K_s\in\R_{\geq 0}$ assumed for the storage function $S:X\rightarrow \R_{\geq 0}$ in As. \ref{As:LipschitzStorageFunction}. To compute this $K_s$ parameter, theoretically, under some assumptions, one can exploit the concept of available storage function $S_a:X\rightarrow \R_{\geq 0}$ defined in Def. \ref{Def:AvailableStorage}. However, here we take a more practical approach as given in the sequel. 

First, by applying $K_s=0$ in \eqref{Eq:Th:EstimatedOptimalSystemIndices}, we obtain parameter-free estimates for the optimal system indices \eqref{Eq:Lm:OptimalSystemIndices}: L2G $\gamma_*$, IFP $\nu_*$ and OFP $\rho_*$ respectively as 
\begin{equation}\label{Eq:EstimatedOptimalSystemIndicesParameterFree}
\begin{gathered}
    \tilde{\gamma}^2 = \max_{t \in [t_o,t_f)\backslash \tilde{\Psi}_u} 
    \frac{y^\T(t)y(t)}{u^\T(t)u(t)},\\   
    \tilde{\nu} = \min_{t \in [t_o,t_f)\backslash \tilde{\Psi}_u} 
    \frac{u^\T(t)y(t)}{u^\T(t)u(t)},\\
    \tilde{\rho} = \min_{t \in [t_o,t_f)\backslash \tilde{\Psi}_y}
    \frac{u^\T(t)y(t)}{y^\T(t)y(t)}.
\end{gathered}
\end{equation} 

The following lemma compares \eqref{Eq:Lm:OptimalSystemIndices}, \eqref{Eq:EstimatedOptimalSystemIndicesOld}, \eqref{Eq:Th:EstimatedOptimalSystemIndices} and \eqref{Eq:EstimatedOptimalSystemIndicesParameterFree}.

\begin{lemma}\label{Lm:Bounds}
The optimal system indices \eqref{Eq:Lm:OptimalSystemIndices} and their on-line estimates given in \eqref{Eq:EstimatedOptimalSystemIndicesOld}, \eqref{Eq:Th:EstimatedOptimalSystemIndices} and \eqref{Eq:EstimatedOptimalSystemIndicesParameterFree} satisfy:
\begin{equation}\label{Eq:Lm:Bounds}
\begin{aligned}
\hat{\gamma}^2_* &\leq \gamma^2_*,  \ \ \ \   & \hat{\nu}_* &\geq \nu_*,  \ \ \ \  & \hat{\rho}_* &\geq \rho_*,\\
\hat{\gamma}^2_* &\leq \tilde{\gamma}^2,            & \hat{\nu}_* &\geq \tilde{\nu},             & \hat{\rho}_* &\geq \tilde{\rho},\\
\hat{\gamma}^2 &\leq \tilde{\gamma}^2,             & \hat{\nu} &\geq \tilde{\nu},             & \hat{\rho} &\geq \tilde{\rho}.
\end{aligned}
\end{equation}
\end{lemma}
\begin{proof}
Consider the first set of inequalities in \eqref{Eq:Lm:Bounds}. In there, note that, the inequality $\hat{\gamma}^2_* \leq \gamma^2_*$ has already been established in Th. \ref{Th:EstimatedOptimalSystemIndices}. Further, due to the fact that $K_s\geq 0$ and $u^\T(t)u(t)>0$, from comparing \eqref{Eq:Th:EstimatedOptimalSystemIndices} with \eqref{Eq:EstimatedOptimalSystemIndicesParameterFree}, it is easy to see that $\hat{\gamma}^2_* \leq \tilde{\gamma}^2$. Furthermore, using: \eqref{Eq:EstimatedOptimalSystemIndicesOld}, $(t_o,t_f)\in\hat{\Psi}_\gamma$, Co. \ref{Co:Th:RationalFunctionPropertyCont2} and \eqref{Eq:EstimatedOptimalSystemIndicesParameterFree}, we can simplify 
\begin{align*}
\hat{\gamma}^2 
\triangleq \frac{\int_{t_o}^{t_f} y^\T(t)y(t)dt}{\int_{t_o}^{t_f}u^\T(t)u(t)dt} 
\leq& \underset{(t_0,t_1)\in\hat{\Psi}_\gamma}{\max} 
\frac{\int_{t_0}^{t_1} y^\T(t)y(t)dt}{\int_{t_0}^{t_1}u^\T(t)u(t)dt}\\
=& \underset{t\in[t_o,t_f)\backslash \tilde{\Psi}_u}{\max} 
\frac{y^\T(t)y(t)}{u^\T(t)u(t)} 
\triangleq \tilde{\gamma}^2,
\end{align*}
and thus, $\hat{\gamma}^2 \leq \tilde{\gamma}^2$. Following the same steps, the last two sets of inequalities in \eqref{Eq:Lm:Bounds} can also be established.
\end{proof}

A critical consequence of Lm. \ref{Lm:Bounds} is that it can be used to establish bounds on the optimal system indices \eqref{Eq:Lm:OptimalSystemIndices} with respect to their estimates proposed in \eqref{Eq:EstimatedOptimalSystemIndicesOld} \cite{Zakeri2019} as: 
\begin{equation}\label{Eq:BoundsForOldEstimates}
    \begin{aligned}
    \hat{\gamma}^2 - 2(\tilde{\gamma}^2-\hat{\gamma}_*^2) \leq \gamma_*^2 \ \ &\mbox{ with } \ \  (\tilde{\gamma}^2-\hat{\gamma}_*^2)\geq0,\\
    \hat{\nu} + 2(\hat{\nu}_*-\tilde{\nu}) \geq \nu_* \ \ &\mbox{ with }\ \ 
    (\hat{\nu}_*-\tilde{\nu})\geq 0,\\
    \hat{\rho} + 2(\hat{\rho}_*-\tilde{\rho}) \geq \rho_* \ \ &\mbox{ with } \ \ 
    (\hat{\rho}_*-\tilde{\rho})\geq 0.
    \end{aligned}
\end{equation}

On the other hand, note that, as $K_s\rightarrow 0$, the optimal estimates given in \eqref{Eq:Th:EstimatedOptimalSystemIndices} may become: (i) inaccurate (due to violation of As. \ref{As:LipschitzStorageFunction}), (ii) close to the optimal values \eqref{Eq:Lm:OptimalSystemIndices}, (iii) close to the estimates given in \eqref{Eq:EstimatedOptimalSystemIndicesParameterFree} and (iv) distant from the estimates given in \eqref{Eq:EstimatedOptimalSystemIndicesOld}. 
On the other hand, as $K_s$ increases, the optimal estimates given in \eqref{Eq:Th:EstimatedOptimalSystemIndices} may become: (i) accurate (due to satisfaction of As. \ref{As:LipschitzStorageFunction}), (ii) distant from the optimal values \eqref{Eq:Lm:OptimalSystemIndices}, (iii) distant from the estimates given in \eqref{Eq:EstimatedOptimalSystemIndicesParameterFree}, and (iv) close to the estimates given in \eqref{Eq:EstimatedOptimalSystemIndicesOld}. Therefore, the importance of selecting a moderate $K_s$ value is clear. 

For this purpose, we assume the availability of a test input-output profile of the system \eqref{Eq:NonlinearSystem} (i.e., $\{(u(t),y(t)):t\in[t_o,t_f)\}$) and the corresponding evaluated estimates given in \eqref{Eq:EstimatedOptimalSystemIndicesOld} (i.e., $\hat{\gamma}^2,\hat{\nu},\hat{\rho}$) and \eqref{Eq:EstimatedOptimalSystemIndicesParameterFree} (i.e., $\tilde{\gamma}^2,\tilde{\nu},\tilde{\rho}$). Let us define 
\begin{equation}\label{Eq:MeanEstimates}
    \bar{\gamma}^2 \triangleq \frac{1}{2}(\hat{\gamma}^2+\tilde{\gamma}^2),\ \ \ \ 
    \bar{\nu} \triangleq \frac{1}{2}(\hat{\nu}+\tilde{\nu}), \ \ \ \ 
    \bar{\rho} \triangleq \frac{1}{2}(\hat{\rho}+\tilde{\rho}),
\end{equation}
and three candidate $K_s$ parameter values:
\begin{equation}\label{Eq:CandidateK_sValues}
    \begin{aligned}
    K_s \triangleq&\ \max_{t\in[t_o,t_f)\backslash \tilde{\Psi}_u} y^\T(t)y(t) - \bar{\gamma}^2 u^\T(t)u(t),\\
    K_s \triangleq&\ \max_{t\in[t_o,t_f)\backslash \tilde{\Psi}_u} \bar{\nu} u^\T(t)u(t) - u^\T(t)y(t),\\
    K_s \triangleq&\ \max_{t\in[t_o,t_f)\backslash \tilde{\Psi}_y} \bar{\rho} y^\T(t)y(t) - u^\T(t)y(t).\\
    \end{aligned}
\end{equation}


\begin{theorem}\label{Th:CandidateK_sValues}
When the three optimization problems in \eqref{Eq:Th:EstimatedOptimalSystemIndices} are solved using the respective $K_s$ parameters given in \eqref{Eq:CandidateK_sValues}, the resulting estimates of the optimal system indices \eqref{Eq:Th:EstimatedOptimalSystemIndices} will lie within the bounds \eqref{Eq:BoundsForOldEstimates} found for the estimates given by \eqref{Eq:EstimatedOptimalSystemIndicesOld} \cite{Zakeri2019} (implying an increased accuracy). 
\end{theorem}
\begin{proof}
We first consider the estimation of the L2G value (using the first equations in \eqref{Eq:Th:EstimatedOptimalSystemIndices} and \eqref{Eq:CandidateK_sValues}). From \eqref{Eq:CandidateK_sValues}, 
\begin{align}
    K_s &\geq y^\T(t)y(t)-\bar{\gamma}^2 u^\T(t)u(t),  \ \ \forall t\in [t_o,t_f)\backslash \tilde{\Psi}_u \nonumber\\
    &\iff \bar{\gamma}^2  \geq \frac{y^\T(t)y(t)-K_s}{u^\T(t)u(t)}, \ \ \forall t\in [t_o,t_f)\backslash \tilde{\Psi}_u \nonumber\\
    &\iff \bar{\gamma}^2 \geq \hat{\gamma}_*^2, \ \ \mbox{(sub. from \eqref{Eq:Th:EstimatedOptimalSystemIndices})}\nonumber\\
    &\iff \frac{1}{2}(\hat{\gamma}^2+\tilde{\gamma}^2) \geq \hat{\gamma}_*^2. \ \ \mbox{(sub. from \eqref{Eq:MeanEstimates})}\label{Eq:Th:CandidateK_sValuesStep1}
\end{align}
According to Lm. \ref{Lm:Bounds}, $\tilde{\gamma}^2 \geq \hat{\gamma}^2$. It is easy see that 
$$\tilde{\gamma}^2 \geq \hat{\gamma}^2 \iff 2\tilde{\gamma}^2-\hat{\gamma}^2 \geq \frac{1}{2}(\hat{\gamma}^2+\tilde{\gamma}^2).$$
Applying this result in \eqref{Eq:Th:CandidateK_sValuesStep1}, we get 
\begin{align}
    2\tilde{\gamma}^2-\hat{\gamma}^2 \geq \hat{\gamma}_*^2 
    &\iff 
    2\tilde{\gamma}^2-2\hat{\gamma}_*^2 -\hat{\gamma}^2 \geq -\hat{\gamma}_*^2 \nonumber\\
    &\iff \hat{\gamma}^2 - 2(\tilde{\gamma}^2-\hat{\gamma}_*^2) \leq \hat{\gamma}_*^2 \nonumber\\
    &\iff \hat{\gamma}^2 - 2(\tilde{\gamma}^2-\hat{\gamma}_*^2) \leq \hat{\gamma}_*^2 \leq \gamma_*^2 
    \ \ \mbox{(Lm. \ref{Lm:Bounds})}\nonumber 
\end{align}
Comparing this result with \eqref{Eq:BoundsForOldEstimates} leads to the conclusion that the estimate $\hat{\gamma}_*^2$ (evaluated using \eqref{Eq:Th:EstimatedOptimalSystemIndices} under \eqref{Eq:CandidateK_sValues}) will fall within the bound given in \eqref{Eq:BoundsForOldEstimates} (for the estimate \eqref{Eq:EstimatedOptimalSystemIndicesOld} \cite{Zakeri2019})
\end{proof}

\section{Generalization for Discrete-Time Systems}\label{Sec:GeneralizationToDiscreteTimeSystems}

So far, we have focused only on continuous-time non-linear dynamical systems \eqref{Eq:NonlinearSystem}. In contrast, this section will focus on discrete-time non-linear dynamical systems and briefly specialize the previously presented results.

\paragraph{\textbf{Discrete-Time System}}  
Consider the discrete-time non-linear system: 
\begin{equation}\label{Eq:NonlinearSystemDisc}
    \mathcal{H}:
    \begin{cases}
    x(t+1) &= f(x(t),u(t)),\\
    y(t) &= h(x(t),u(t)),
    \end{cases}
\end{equation}
where $t\in\Z_{\geq 0}$. It is easy to see that As. \ref{As:Dynamics}, Defs. \ref{Def:SupplyRate}-\ref{Def:OFPIndex}, Rm. \ref{Rm:ActualSystemIndices}, Lm. \ref{Lm:OptimalSystemIndices}, \eqref{Eq:EstimatedOptimalSystemIndicesOld} and Lm. \ref{Lm:EstimatedOptimalSystemIndices} are also valid with regard to the discrete-time system \eqref{Eq:NonlinearSystemDisc} upon some obvious minor modifications (e.g., replacing continuous integrals with discrete summations). Hence, we omit re-stating them here in the interest of brevity. Note however that As. \ref{As:LipschitzStorageFunction} now requires a special modification, and thus, we re-state it as follows.  
\begin{assumption}
The magnitude change in the storage function $S:X\rightarrow \R_{\geq 0}$ over a single time step is bounded such that $\vert S(x(t+1))-S(x(t)) \vert \leq K_s$ over the discrete-time period $t\in [t_o,t_f)$, where $K_s$ is a known constant.
\end{assumption}

\paragraph{\textbf{Fractional Function Optimization}}
Considering the discrete-time case, we will now specialize some key results established in Sec. \ref{Sec:FractionalFunctionOptimization}. First, let us define the discrete-time version of the generic fractional function form given in \eqref{Eq:FractionalFunctionsOfInterest}. Let $t_o,t_f\in\Z_{\geq 0}$ where $t_o<t_f$ be two discrete-time instants and $p(t)$,$q(t)$ where $p:[t_o,t_f) \rightarrow \R$, $q:[t_o,t_f)\rightarrow \R_{\geq 0}$ be two discrete-time signals. Moreover, let $r(\tau)$, $s(\tau)$ where $r:\Z_{>0}\rightarrow \R$, $s:\Z_{>0}\rightarrow \R_{\geq 0}$ be two discrete linear mappings.

Now, for any discrete-time interval $[t_0,t_1) \subseteq [t_o,t_f)$, let us define a corresponding \emph{fractional function} of the form:
\begin{equation}\label{Eq:GenericRationalFunctionDisc}
    \psi(t_0,t_1) \triangleq
    \frac{\sum_{t\in [t_0,t_1)} p(t) + r(t_1-t_0)}{\sum_{t\in[t_0,t_1)} q(t) + s(t_1-t_0)}.
\end{equation}
Note that As. \ref{As:NoDeadZoneCont}, Lm. \ref{Lm:RationalFunctionPropertyCont} and the set $\Psi$ \eqref{Eq:GenericRationalFunctionContSpace} are also valid for any fractional function $\psi(t_0,t_1)$ of the form \eqref{Eq:GenericRationalFunctionDisc} as well. Thus, parallel to Th. \ref{Th:RationalFunctionPropertyCont}, the following theorem can be proved.

\begin{table*}[!h]
\hrule 
\begin{align}\label{Eq:Co:RationalFunctionPropertyMaxDisc}
    &\Big\{ 
    \underset{(t_0,t_1) \in \Psi \backslash \Psi_q }{\max}\ 
    \psi(t_0,t_1)  
    \Big\} 
    \equiv\ 
    \max 
    \Big\{
    \big\{
    \underset{t \in [t_o,t_f)\backslash \tilde{\Psi}_q}{\max}\ 
    \psi(t,t+1) 
    \big\},\ 
    \big\{
    \underset{\substack{t\in [t_{zo},t_{zf})\\(t_{zo},t_{zf})\in\bar{\Psi}_q}}{\max}\ 
    \psi(t_{zo}-1,t)
    \big\},\ 
    \big\{
    \underset{\substack{t\in [t_{zo},t_{zf})\\(t_{zo},t_{zf})\in\bar{\Psi}_q}}{\max}\ 
    \psi(t,t_{zf}+1) 
    \big\}
    \Big\},
    \\
    \label{Eq:Co:RationalFunctionPropertyMinDisc}
    &\Big\{ 
    \underset{(t_0,t_1)\in\Psi\backslash \Psi_q}{\min}\ 
    \psi(t_0,t_1)  
    \Big\} 
    \equiv\ 
    \min 
    \Big\{
    \big\{
    \underset{t \in [t_o,t_f)\backslash \tilde{\Psi}_q}{\min}\ 
    \psi(t,t+1) 
    \big\},\ 
    \big\{
    \underset{\substack{t\in [t_{zo},t_{zf})\\(t_{zo},t_{zf})\in\bar{\Psi}_q}}{\min}\ 
    \psi(t_{zo}-1,t)
    \big\},\  
    \big\{
    \underset{\substack{t\in [t_{zo},t_{zf})\\(t_{zo},t_{zf})\in\bar{\Psi}_q}}{\min}\ 
    \psi(t,t_{zf}+1)
    \big\}
    \Big\},
    \\ \label{Eq:Co:RationalFunctionPropertyObjectives1Disc}
    &\psi(t_0,t_1) 
    \triangleq\  \frac{\sum_{t\in[t_0,t_1)} p(t) + r(t_1-t_0)}{\sum_{t\in[t_0,t_1)} q(t)}, 
    \ \ \ \mbox{(i.e., \eqref{Eq:GenericRationalFunctionDisc} under As. \ref{As:NoLinearMappingInDenominator})\ \ \  and }\ \ \ 
    \psi(t,t+1) 
    = \frac{p(t)+r(1)}{q(t)},\\ \label{Eq:Co:RationalFunctionPropertyObjectives2Disc}
    &\psi(t_{zo}-1,t) 
    =   
    \frac{\sum_{\tau \in[t_{zo}-1,t)} p(\tau) + r(t-t_{zo}+1)}{q(t_{zo}-1)}, \ \ \ 
    \mbox{ and }\ \ \ 
    \psi(t,t_{zf}+1) 
    =
    \frac{\sum_{\tau \in[t,t_{zf}+1)} p(\tau) + r(t_{zf}+1-t)}{q(t_{zf})}.
\end{align}
\hrule
\end{table*}

\begin{theorem}\label{Th:RationalFunctionPropertyDisc}
Any fractional function $\psi(t_0,t_1)$ of the form \eqref{Eq:GenericRationalFunctionDisc} under As. \ref{As:NoDeadZoneCont} can be efficiently optimized over the corresponding space $(t_0,t_1)\in\Psi$ exploiting the relationships:
\begin{equation}\label{Eq:Th:RationalFunctionPropertyMaxDisc}
    \Big\{ 
    \underset{(t_0,t_1)\in\Psi}{\max}\ 
    \psi(t_0,t_1)  
    \Big\} 
    \equiv
    \Big\{
    \underset{t\in [t_o,t_f)}{\max}\ 
    \psi(t,t+1) 
    \Big\},
\end{equation}
\begin{equation}\label{Eq:Th:RationalFunctionPropertyMinDisc}
    \Big\{ 
    \underset{(t_0,t_1)\in\Psi}{\min}\ 
    \psi(t_0,t_1) 
    \Big\} 
    \equiv
    \Big\{
    \underset{t\in [t_o,t_f)}{\min}\ 
    \psi(t,t+1) 
    \Big\},
\end{equation}
where 
\begin{equation*}
    \psi(t,t+1) = \frac{p(t)+r(1)}{q(t)+s(1)}.
\end{equation*}
\end{theorem}
\begin{proof}
The proof follows the same steps as that of Th. \ref{Th:RationalFunctionPropertyCont} and is, therefore, omitted.
\end{proof}

Note that unlike in Th. \ref{Th:RationalFunctionPropertyCont}, Th. \ref{Th:RationalFunctionPropertyDisc} does not involve $\dot{r}(0)$ and $\dot{s}(0)$ values and limits $\lim_{\Delta\rightarrow 0}\psi(t,t+\Delta)$. Instead, it now involves $r(1)$ and $s(1)$ values and unit steps $\psi(t,t+1)$.

We next relax As. \ref{As:NoDeadZoneCont} in Th. \ref{Th:RationalFunctionPropertyDisc} by following the same steps as before, i.e., by making a simplifying assumption (As. \ref{As:NoLinearMappingInDenominator}) and defining sets $\Psi_q$, $\bar{\Psi}_q$ and $\tilde{\Psi}_q$ appropriately. Note that As. \ref{As:NoLinearMappingInDenominator} and the sets $\Psi_q$, $\bar{\Psi}_q$ (with $\epsilon=1$) and $\tilde{\Psi}_q$ defined earlier are also valid for any fractional function $\psi(t_0,t_1)$ of the form \eqref{Eq:GenericRationalFunctionDisc} as well. Therefore, parallel to Th. \ref{Th:Th:RationalFunctionPropertyCont}, the following theorem now can be established.

\begin{theorem}\label{Th:Th:RationalFunctionPropertyDisc}
Any fractional function $\psi(t_0,t_1)$ of the form \eqref{Eq:GenericRationalFunctionDisc} under As. \ref{As:NoLinearMappingInDenominator} can be efficiently optimized over the corresponding space $(t_0,t_1)\in\Psi\backslash \Psi_q$ exploiting the relationships: 
\eqref{Eq:Co:RationalFunctionPropertyMaxDisc} and \eqref{Eq:Co:RationalFunctionPropertyMinDisc} (under the objective functions in  \eqref{Eq:Co:RationalFunctionPropertyObjectives1Disc} and \eqref{Eq:Co:RationalFunctionPropertyObjectives2Disc}). 
\end{theorem}
\begin{proof}
The proof follows the same steps as that of Th. \ref{Th:Th:RationalFunctionPropertyCont} and is, therefore, omitted.
\end{proof}

In parallel to Co. \ref{Co:Th:RationalFunctionPropertyCont2}, the following corollary shows that under some special conditions, when optimizing a fractional function of the form \eqref{Eq:GenericRationalFunctionDisc}, the dead-zones can be conveniently disregarded. Before proceeding, we point out that As. \ref{As:NumeratorDenominatorConnection} is also valid (upon some obvious minor modifications) for any fractional function of the form \eqref{Eq:GenericRationalFunctionDisc} as well. 

\begin{corollary}\label{Co:Th:RationalFunctionPropertyDisc2}
Any fractional function $\psi(t_0,t_1)$ of the form \eqref{Eq:GenericRationalFunctionDisc} under As. \ref{As:NoLinearMappingInDenominator} and As. \ref{As:NumeratorDenominatorConnection} can be efficiently optimized over the corresponding space $(t_0,t_1)\in\Psi\backslash\Psi_q$ exploiting the relationships: 
(i) If $r(\tau) \leq 0, \forall \tau \in\Z_{>0}$: 
\begin{equation}\label{Eq:Co:Th:RationalFunctionPropertyDisc2Max}
    \Big\{ 
    \underset{(t_0,t_1) \in \Psi\backslash\Psi_q}{\max}\ 
    \psi(t_0,t_1) 
    \Big\} 
    \equiv
    \Big\{
    \underset{t \in [t_o,t_f)\backslash \tilde{\Psi}_q}{\max}\ 
    \psi(t,t+1) 
    \Big\},
\end{equation}
(ii) If $r(\tau) \geq 0, \forall \tau \in\Z_{>0}$:
\begin{equation}\label{Eq:Co:Th:RationalFunctionPropertyDisc2Min}
    \Big\{ 
    \underset{(t_0,t_1) \in \Psi\backslash\Psi_q}{\min}\ 
    \psi(t_0,t_1) 
    \Big\} 
    \equiv
    \Big\{
    \underset{t \in [t_o,t_f)\backslash \tilde{\Psi}_q}{\min}\ 
    \psi(t,t+1) 
    \Big\},
\end{equation}
where 
\begin{equation*}
    \psi(t,t+1) = \frac{p(t)+r(1)}{q(t)}.
\end{equation*}
\end{corollary}
\begin{proof}
These results directly follow from applying the assumed conditions in As. \ref{As:NumeratorDenominatorConnection} and Co. \ref{Co:Th:RationalFunctionPropertyDisc2} in Th. \ref{Th:Th:RationalFunctionPropertyDisc}.
\end{proof}

\paragraph{\textbf{Application to Solve the Discrete Version of \eqref{Eq:Lm:EstimatedOptimalSystemIndices}}}
For the discrete-time system \eqref{Eq:NonlinearSystemDisc}, recall that we had to specialize several theoretical results such as Th. \ref{Th:RationalFunctionPropertyCont}, Th. \ref{Th:Th:RationalFunctionPropertyCont} and Co. \ref{Co:Th:RationalFunctionPropertyCont2} as Th. \ref{Th:RationalFunctionPropertyDisc}, Th. \ref{Th:Th:RationalFunctionPropertyDisc} and Co. \ref{Co:Th:RationalFunctionPropertyDisc2}, respectively. However, when it comes to the application, i.e., solving the optimization problems in \eqref{Eq:Lm:EstimatedOptimalSystemIndices} considering the discrete-time system \eqref{Eq:NonlinearSystemDisc}, it is easy to see that Th. \ref{Th:EstimatedOptimalSystemIndices} can still be used (i.e., valid - in particular, in the light of Co. \ref{Co:Th:RationalFunctionPropertyDisc2}).

In essence, Th. \ref{Th:EstimatedOptimalSystemIndices} can be used to estimate the optimal system indices of both continuous-time \eqref{Eq:NonlinearSystem} as well as discrete-time \eqref{Eq:NonlinearSystemDisc} systems. Moreover, the same can be said regarding the Remarks \ref{Rm:Efficiency}-\ref{Rm:Realizations} and the theoretical results reported in Lm. \ref{Lm:Bounds} and Th. \ref{Th:EstimatedOptimalSystemIndices}. Therefore, we omit re-stating these results here in the interest of brevity.


\section{Numerical Results}
\label{Sec:NumericalResults}

In our numerical experiments, we considered two linear systems $\mathcal{H}_1, \mathcal{H}_2$ and two non-linear systems $\mathcal{H}_3, \mathcal{H}_4$: 
\begin{equation}\label{Eq:ExampledSystems}
\begin{aligned}
    \mathcal{H}_1:&
    \begin{cases}
    \dot{x} = 
    \begin{bmatrix}
    -0.261 & 1.027 & -0.074\\
   -0.8786 & -0.184 & -0.644\\
    0.537 &  0.364 & -0.576
    \end{bmatrix}x + 
    \begin{bmatrix}
    0 \\ 0.936 \\ 0
    \end{bmatrix}u, \\
    y = 
    \begin{bmatrix}
     -3.020 & -1.103 & -1.032
    \end{bmatrix} x, 
    \end{cases}\\
    \mathcal{H}_2:&
    \begin{cases}
    \dot{x} = 
    \begin{bmatrix}
    -0.681  & 0.495  &-0.651 \\
     0.107  &-0.815  &-0.686 \\
     0.811  & 0.487  &-0.466
    \end{bmatrix}x + 
    \begin{bmatrix}
    0.133\\
    0\\
    -1.399
    \end{bmatrix}u, \\
    y = 
    \begin{bmatrix}
    0 & 0 & 1.025
    \end{bmatrix} x +
    \begin{bmatrix}
    -0.380
    \end{bmatrix}u,
    \end{cases}\\
    \mathcal{H}_3:&
    \begin{cases}
    \dot{x} = 
    -x^3 -x -u, \\
    y = 2x + u,
    \end{cases} \\
    \mathcal{H}_4:& 
    \begin{cases}
    \dot{x} = 
    -x^3 -\frac{1}{3}x -\frac{2}{3} u, \\
    y = \frac{2}{3}x + \frac{4}{3}u.
    \end{cases}  
\end{aligned}
\end{equation}
The corresponding optimal system indices \eqref{Eq:Lm:OptimalSystemIndices} (evaluated analytically in off-line using MATLAB\tsup{\textregistered} and \cite{Xia2018}) are given in Tab. \ref{Tab:NumericalResults} (Col. 3). 
In Tab. \ref{Tab:NumericalResults}, note that we have omitted system index types with infinite or unknown (not established theoretically) optimal index values (e.g., optimal OFP index of $\mathcal{H}_1$ is $-\infty$ and optimal IFP index of $\mathcal{H}_3$ is unknown).

\begin{table*}[!h]
\caption{The summary of the observed numerical results.}
\label{Tab:NumericalResults}
\resizebox{\textwidth}{!}{
\begin{tabular}{|c|l|r|r|rr|rr|r|rr|r|}
\hline
\multicolumn{1}{|c|}{\multirow{2}{*}{System}} &
  \multicolumn{1}{c|}{\multirow{2}{*}{\begin{tabular}[c]{@{}c@{}}System \\ Index\end{tabular}}} &
  \multicolumn{1}{c|}{\multirow{2}{*}{\begin{tabular}[c]{@{}c@{}}Optimal \\ Index Val.\end{tabular}}} &
  \multicolumn{1}{c|}{\multirow{2}{*}{\begin{tabular}[c]{@{}l@{}}Learned \\ $K_s$ Value \end{tabular}}} &
  \multicolumn{2}{c|}{Estimate At $t=100$} &
  \multicolumn{2}{c|}{Abs. Est. Error} &
  \multicolumn{1}{c|}{\multirow{2}{*}{\begin{tabular}[c]{@{}c@{}}\textbf{$\%$ AEE} \\ \textbf{Improvement}\end{tabular}}} &
  \multicolumn{2}{c|}{Mean Abs. Est. Error} &
  \multicolumn{1}{c|}{\multirow{2}{*}{\begin{tabular}[c]{@{}c@{}}\textbf{$\%$ MAEE}  \\ \textbf{Improvement}\end{tabular}}} \\ \cline{5-8} \cline{10-11}
\multicolumn{1}{|c|}{} &
  \multicolumn{1}{c|}{} &
  \multicolumn{1}{c|}{} &
  \multicolumn{1}{c|}{} &
  \multicolumn{1}{c|}{AVG \cite{Zakeri2019}} &
  \multicolumn{1}{c|}{FFO} &
  \multicolumn{1}{c|}{AVG \cite{Zakeri2019}} &
  \multicolumn{1}{c|}{FFO} &
  \multicolumn{1}{c|}{} &
  \multicolumn{1}{c|}{AVG \cite{Zakeri2019}} &
  \multicolumn{1}{c|}{FFO} &
  \multicolumn{1}{c|}{} \\ \hline
\multirow{2}{*}{$\mathcal{H}_1$} &
  L2G &
  17.575 &
  2165.0 &
  \multicolumn{1}{r|}{7.933} &
  16.931 &
  \multicolumn{1}{r|}{9.642} &
  0.644 &
  \textbf{93.32} &
  \multicolumn{1}{r|}{9.344} &
  1.003 &
  \textbf{89.27} \\ \cline{2-12} 
 &
  IFP &
  -8.067 &
  322.5 &
  \multicolumn{1}{r|}{-2.629} &
  -3.850 &
  \multicolumn{1}{r|}{5.438} &
  4.217 &
  \textbf{22.45} &
  \multicolumn{1}{r|}{5.455} &
  4.274 &
  \textbf{21.65} \\ \hline
\multirow{2}{*}{$\mathcal{H}_2$} &
  IFP &
  -2.017 &
  35.06 &
  \multicolumn{1}{r|}{-0.955} &
  -1.276 &
  \multicolumn{1}{r|}{1.062} &
  0.741 &
  \textbf{30.23} &
  \multicolumn{1}{r|}{1.040} &
  0.748 &
  \textbf{28.08} \\ \cline{2-12} 
 &
  OFP &
  -2.630 &
  6.634 &
  \multicolumn{1}{r|}{-0.983} &
  -2.231 &
  \multicolumn{1}{r|}{1.647} &
  0.399 &
  \textbf{75.77} &
  \multicolumn{1}{r|}{1.672} &
  0.399 &
  \textbf{76.14} \\ \hline
$\mathcal{H}_3$ &
  L2G &
  1.000 &
  12.37 &
  \multicolumn{1}{r|}{0.249} &
  0.627 &
  \multicolumn{1}{r|}{0.751} &
  0.373 &
  \textbf{50.33} &
  \multicolumn{1}{r|}{0.742} &
  0.373 &
  \textbf{49.73} \\ \hline
$\mathcal{H}_4$ &
  OFP &
  0.750 &
  2.876 &
  \multicolumn{1}{r|}{0.850} &
  0.799 &
  \multicolumn{1}{r|}{0.100} &
  0.049 &
  \textbf{50.95} &
  \multicolumn{1}{r|}{0.098} &
  0.049 &
  \textbf{50.20} \\ \hline
\end{tabular}}
\end{table*}

To synthetically generate input-output data (on-line), we used Simulink\tsup{\textregistered}. In particular, we set each system initial state as $x(0)=\0$ and generated the output $y(t)$ for the input $u(t)$:
\begin{equation}\label{Eq:InputSignal}
    u(t) = a + cos(bt) + pul(2t) + 0.01v(t), \ \forall t\in[0,100),  
\end{equation}
where $v(t)$ represents a normal random process and $pul(t)$ represents a pulse wave \cite{Oppenheim1997} with $50\%$ duty-cycle, unit period and unit amplitude. For the considered four systems in \eqref{Eq:ExampledSystems}, the parameters $a$ and $b$ in \eqref{Eq:InputSignal} were selected as $16.71,9.71,4.71,4.71$ and $1.02,0.96,0.1,0.1$, respectively. The observed input-output profiles are shown in Fig. \ref{Fig:InputOutputProfiles}.

\begin{figure}[!h]
    \centering
    \begin{subfigure}[t]{0.23\textwidth}
        \centering
        \captionsetup{justification=centering}
        \includegraphics[width=1.6in]{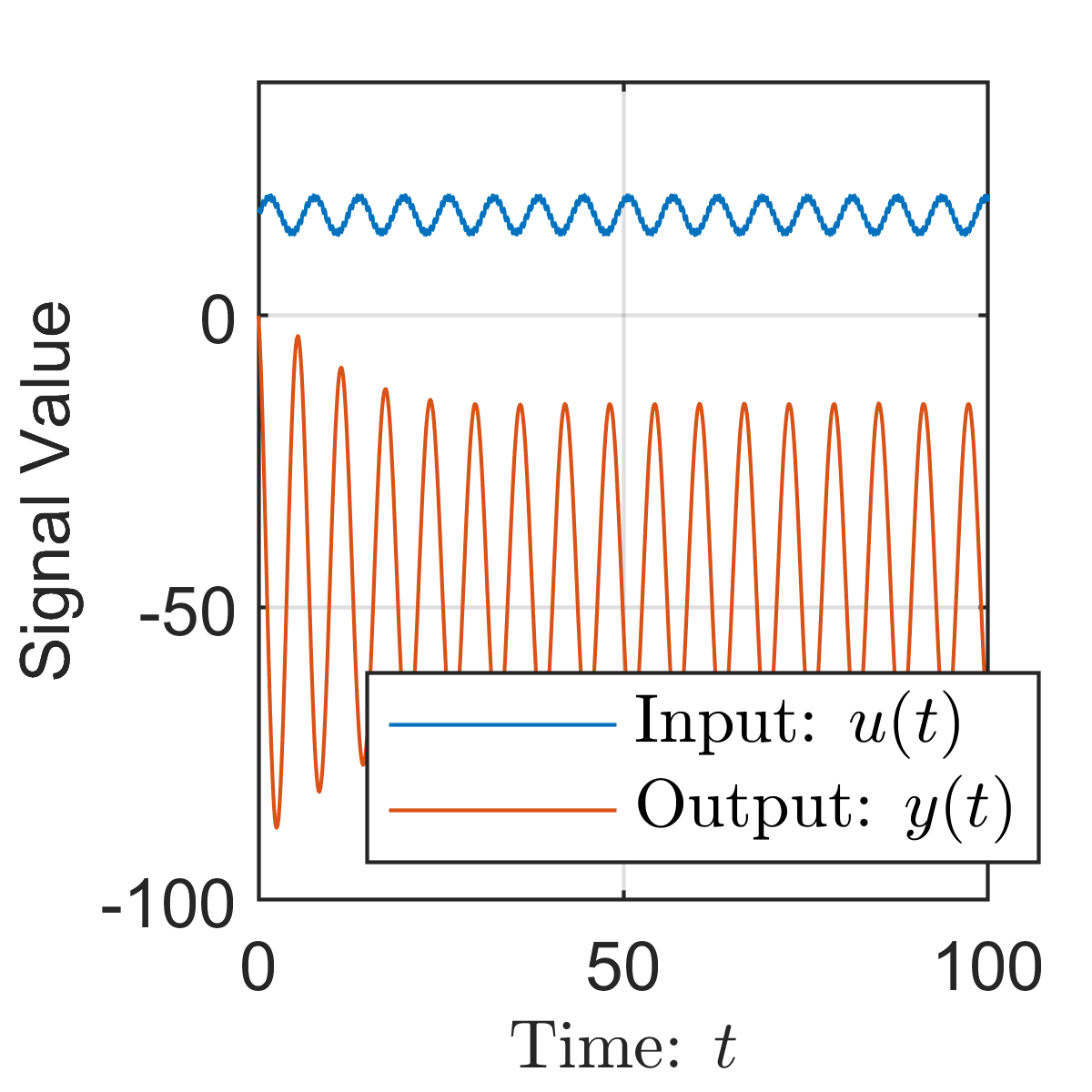}
        \caption{System: $\mathcal{H}_1$}
    \end{subfigure}
    \hfill
    \begin{subfigure}[t]{0.23\textwidth}
        \centering
        \captionsetup{justification=centering}
        \includegraphics[width=1.6in]{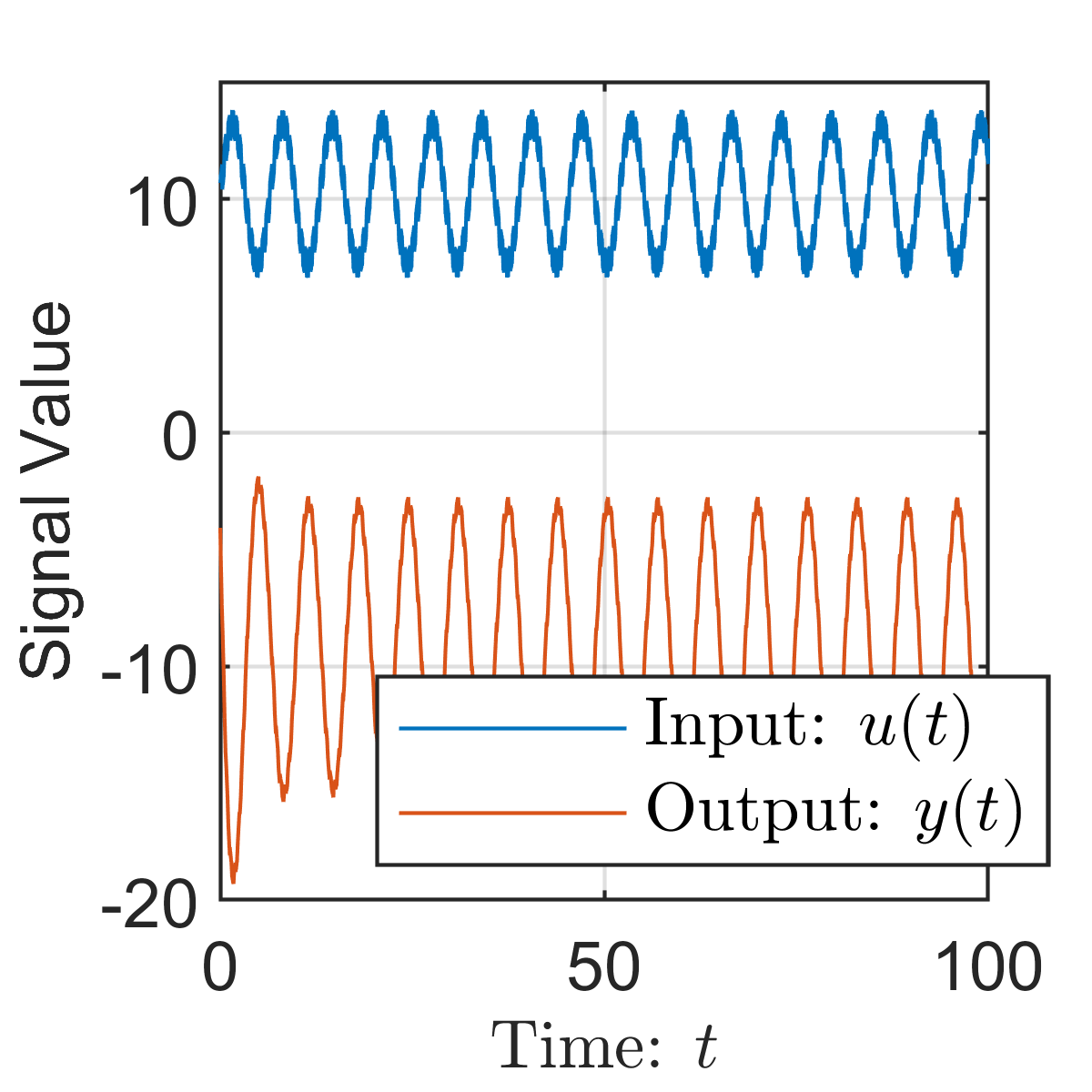}
        \caption{System: $\mathcal{H}_2$}
    \end{subfigure}\\
    \begin{subfigure}[t]{0.23\textwidth}
        \centering
        \captionsetup{justification=centering}
        \includegraphics[width=1.6in]{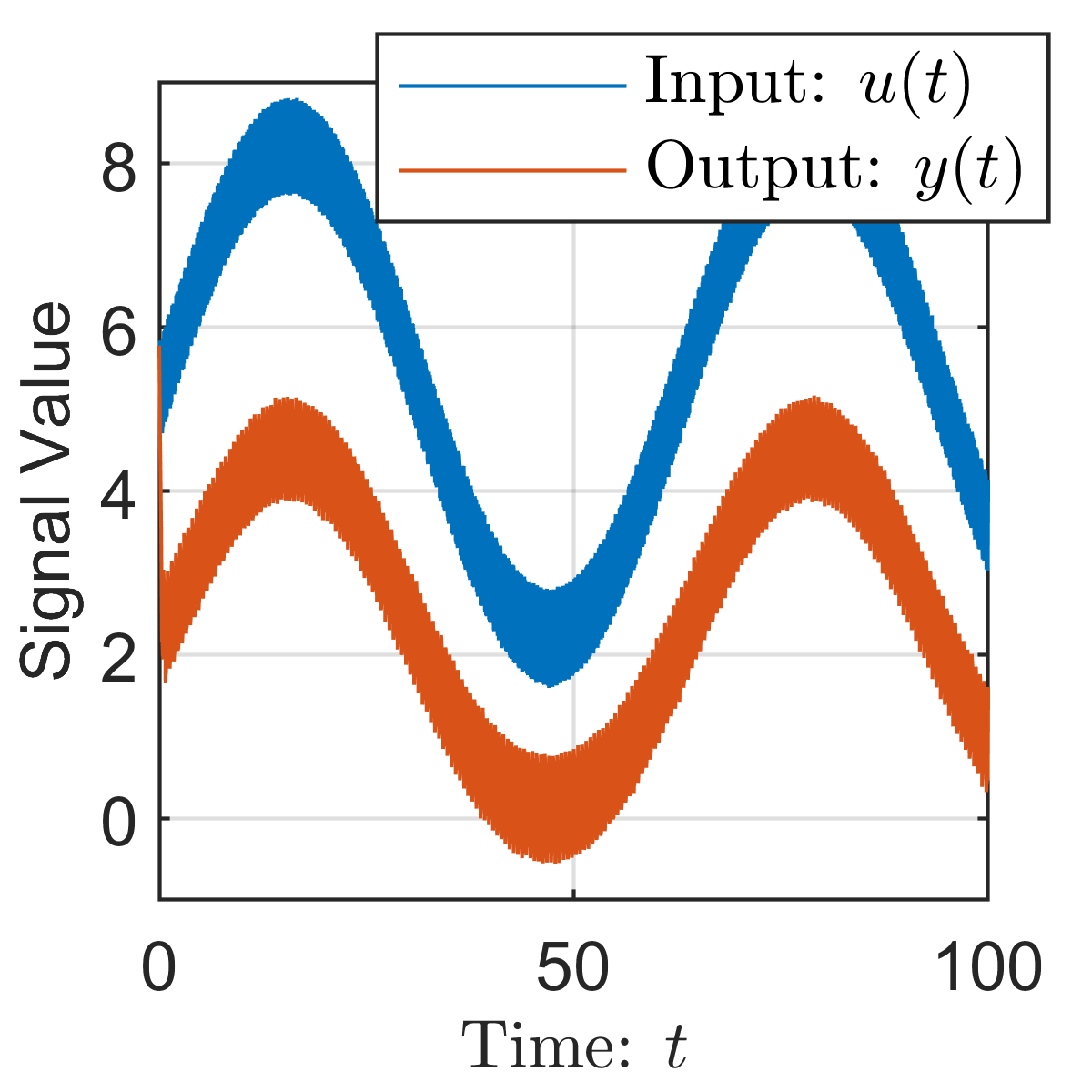}
        \caption{System: $\mathcal{H}_3$}
    \end{subfigure}
    \hfill
    \begin{subfigure}[t]{0.23\textwidth}
        \centering
        \captionsetup{justification=centering}
        \includegraphics[width=1.6in]{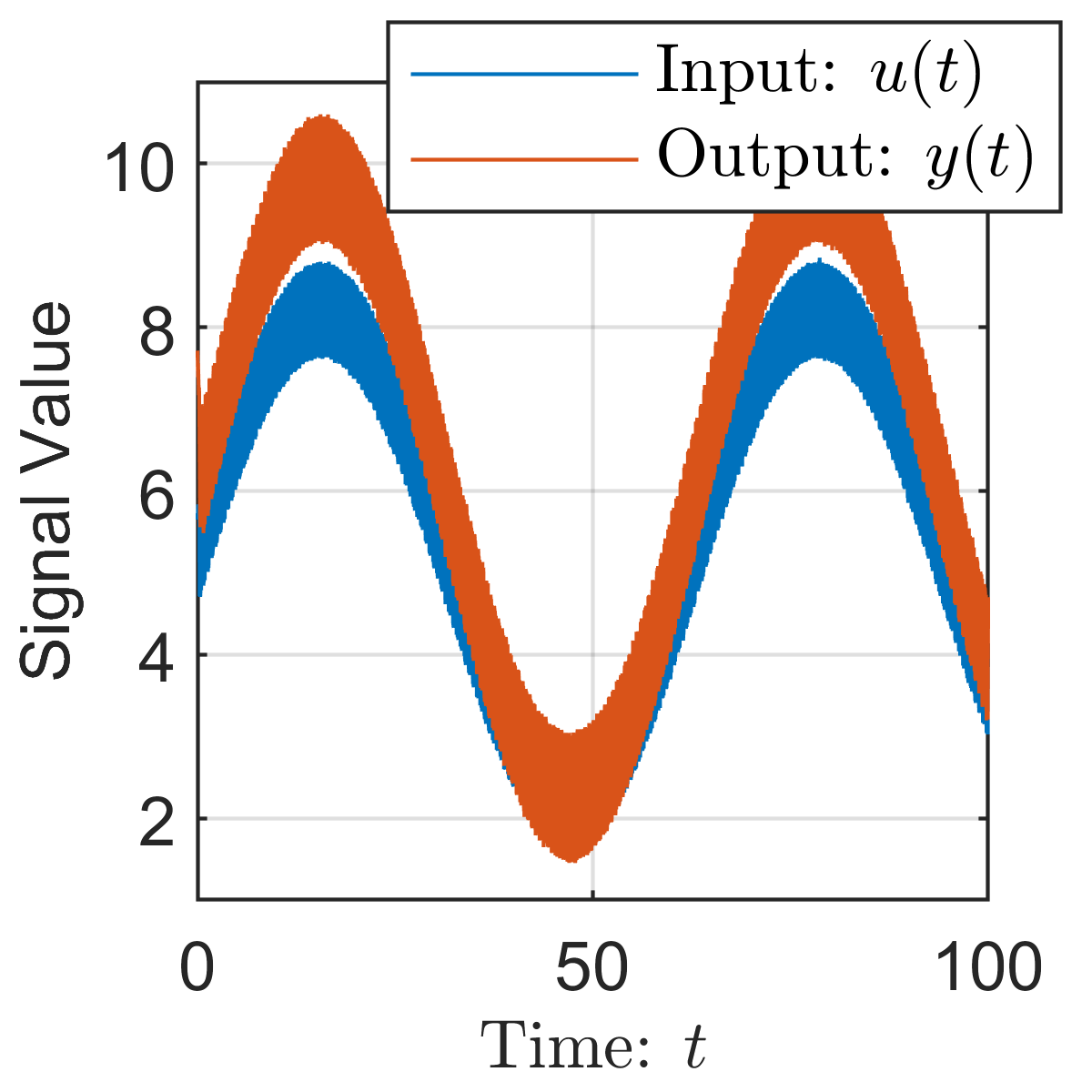}
        \caption{System: $\mathcal{H}_4$}
    \end{subfigure}
    \caption{Observed input-output profiles.}
    \label{Fig:InputOutputProfiles}
\end{figure}

In each experiment, inspired by Th. \ref{Th:CandidateK_sValues}, the input-output data seen over the period $t\in [0,10)$ (sampled at a rate $1\,kHz$) was used to compute/learn the $K_s$ parameter value via \eqref{Eq:CandidateK_sValues}. The learned $K_s$ values are given in Tab. \ref{Tab:NumericalResults} (Col. 4).

In Simulink\tsup{\textregistered}, parallel to the systems in \eqref{Eq:ExampledSystems}, we also implemented the on-line system indices \emph{estimators} described by: 
(i) the averaging-based approach \eqref{Eq:EstimatedOptimalSystemIndicesOld} proposed in \cite{Zakeri2019} (labeled as ``AVG'') and 
(ii) the FFOP-based approach \eqref{Eq:Th:EstimatedOptimalSystemIndices} proposed in this paper (labeled as ``FFO''). Figure \ref{Fig:ObservedEstimates} shows the variations of these observed on-line estimates of different interested system indices over the period $t\in[0,100)$.

\begin{figure}[!h]
    \centering
    \begin{subfigure}[t]{0.23\textwidth}
        \centering
        \captionsetup{justification=centering}
        \includegraphics[width=1.6in]{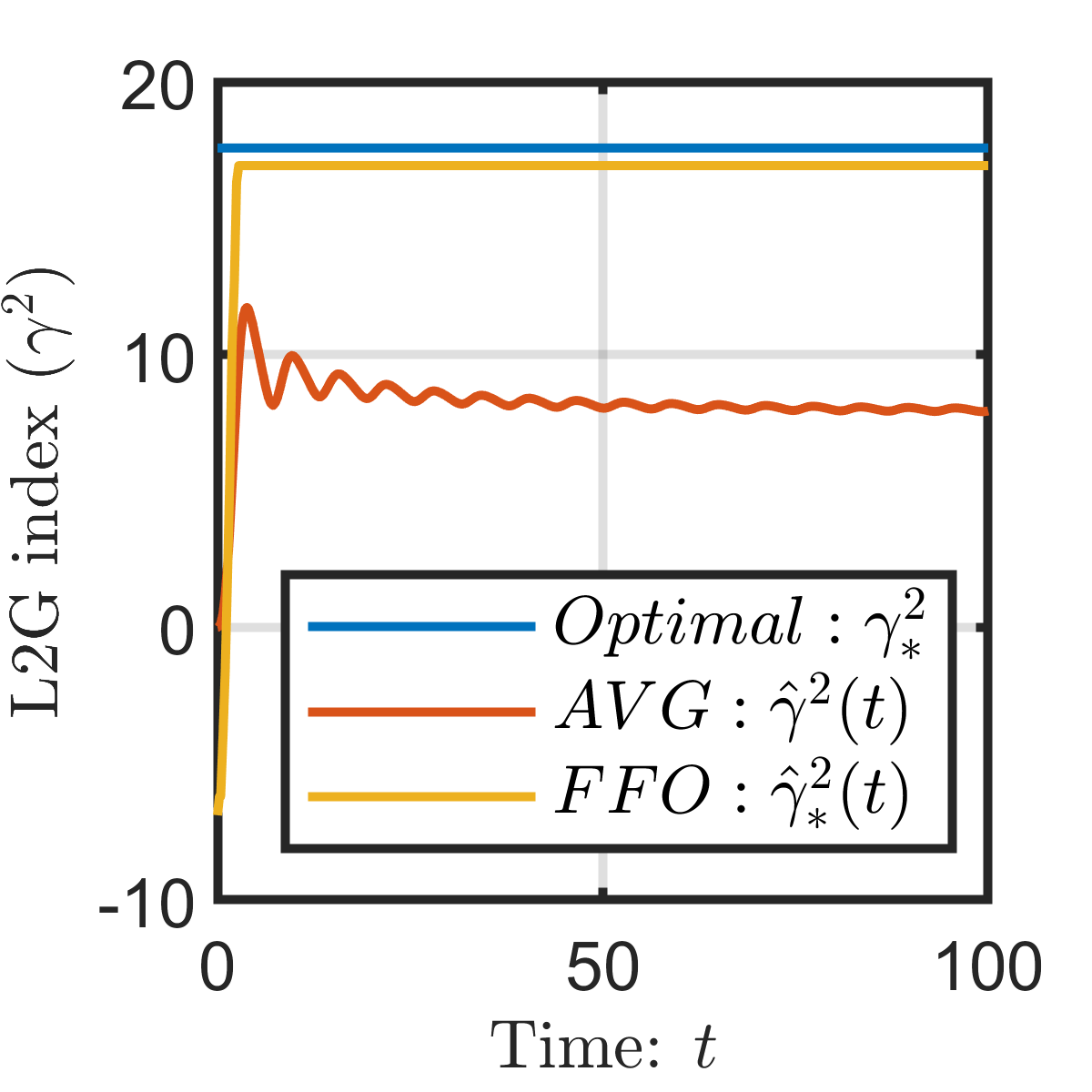}
        \caption{System: $\mathcal{H}_1$, Index: L2G}
    \end{subfigure}
    \hfill
    \begin{subfigure}[t]{0.23\textwidth}
        \centering
        \captionsetup{justification=centering}
        \includegraphics[width=1.6in]{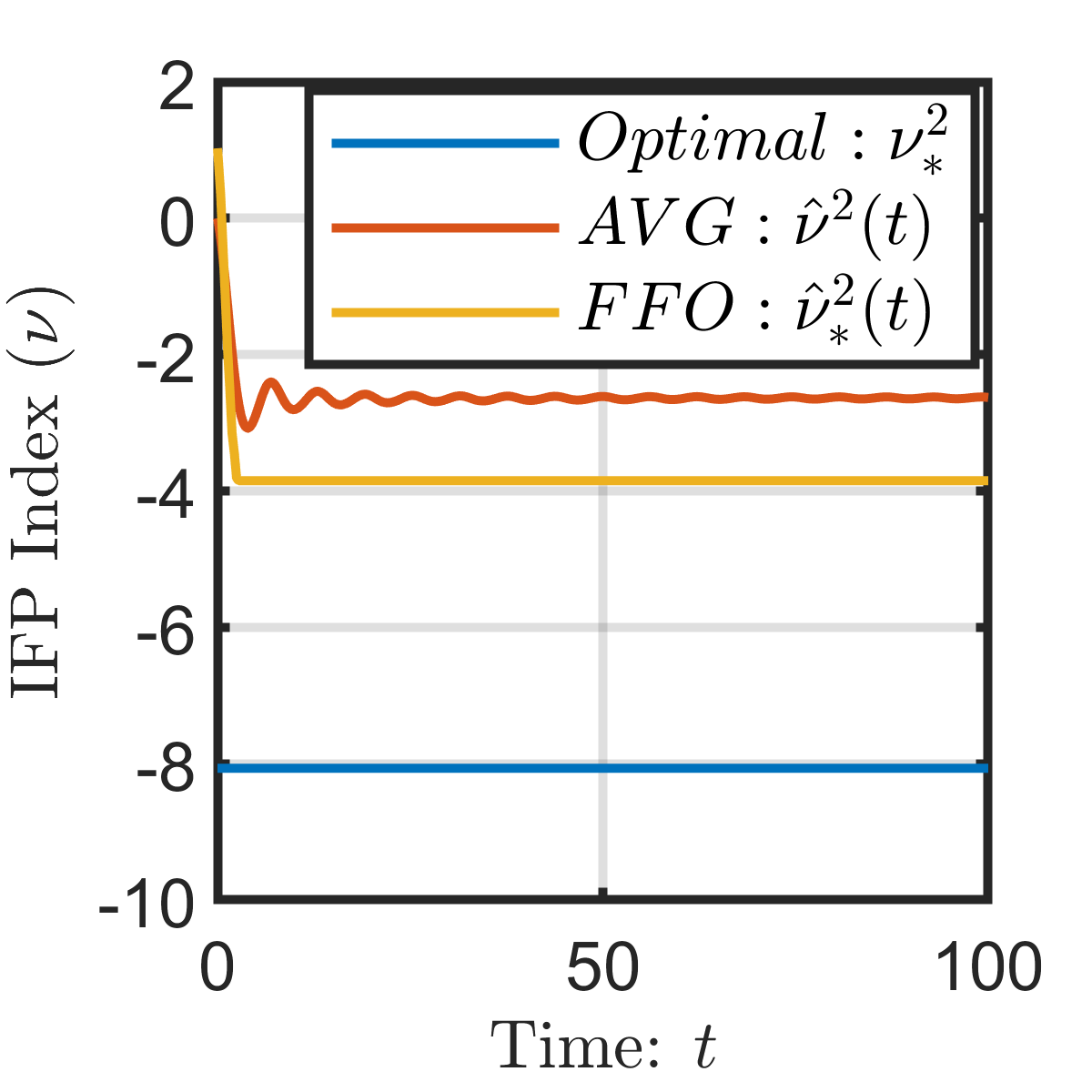}
        \caption{System: $\mathcal{H}_1$, Index: IFP}
    \end{subfigure}\\
    \begin{subfigure}[t]{0.23\textwidth}
        \centering
        \captionsetup{justification=centering}
        \includegraphics[width=1.6in]{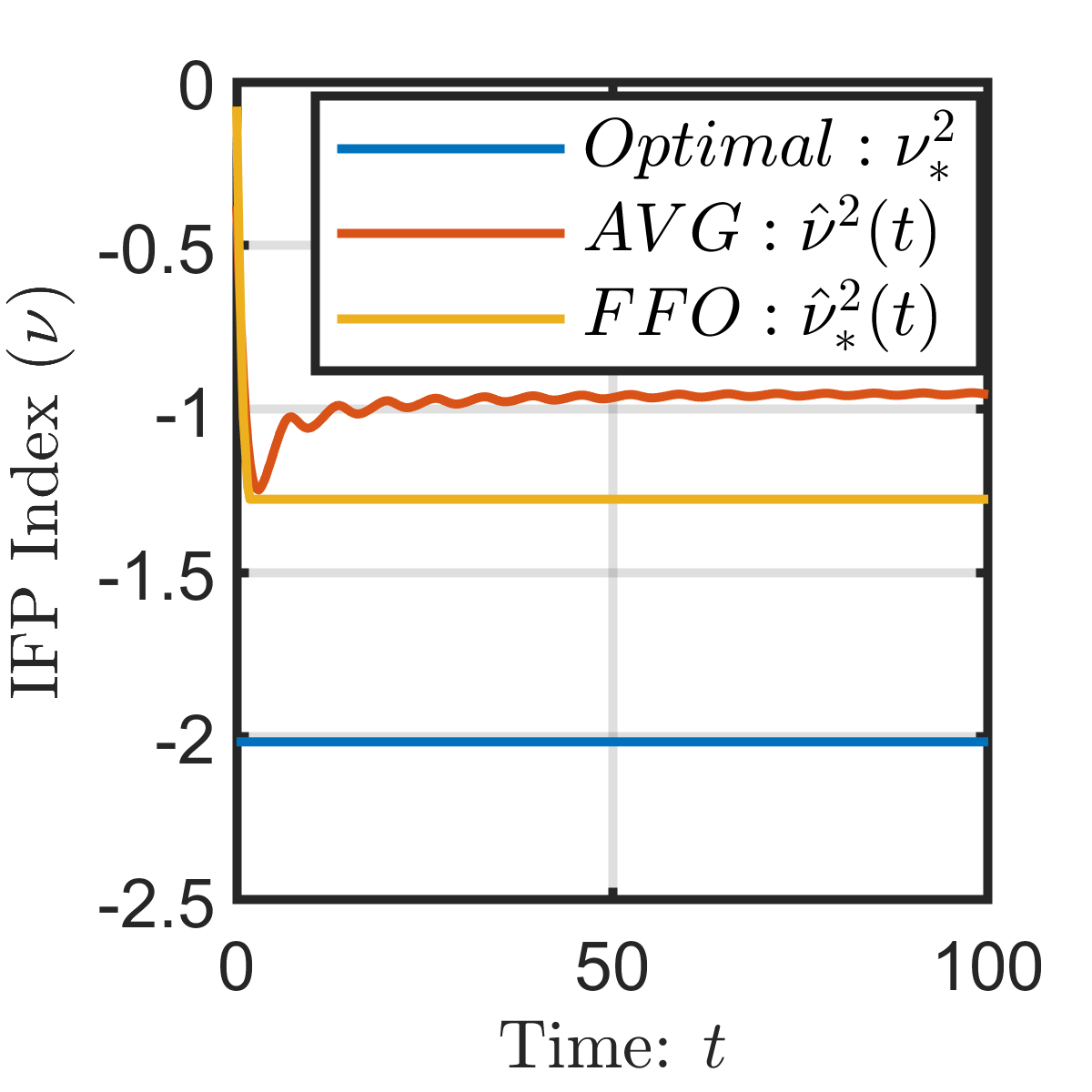}
        \caption{System: $\mathcal{H}_2$, Index: IFP}
    \end{subfigure}
    \hfill
    \begin{subfigure}[t]{0.23\textwidth}
        \centering
        \captionsetup{justification=centering}
        \includegraphics[width=1.6in]{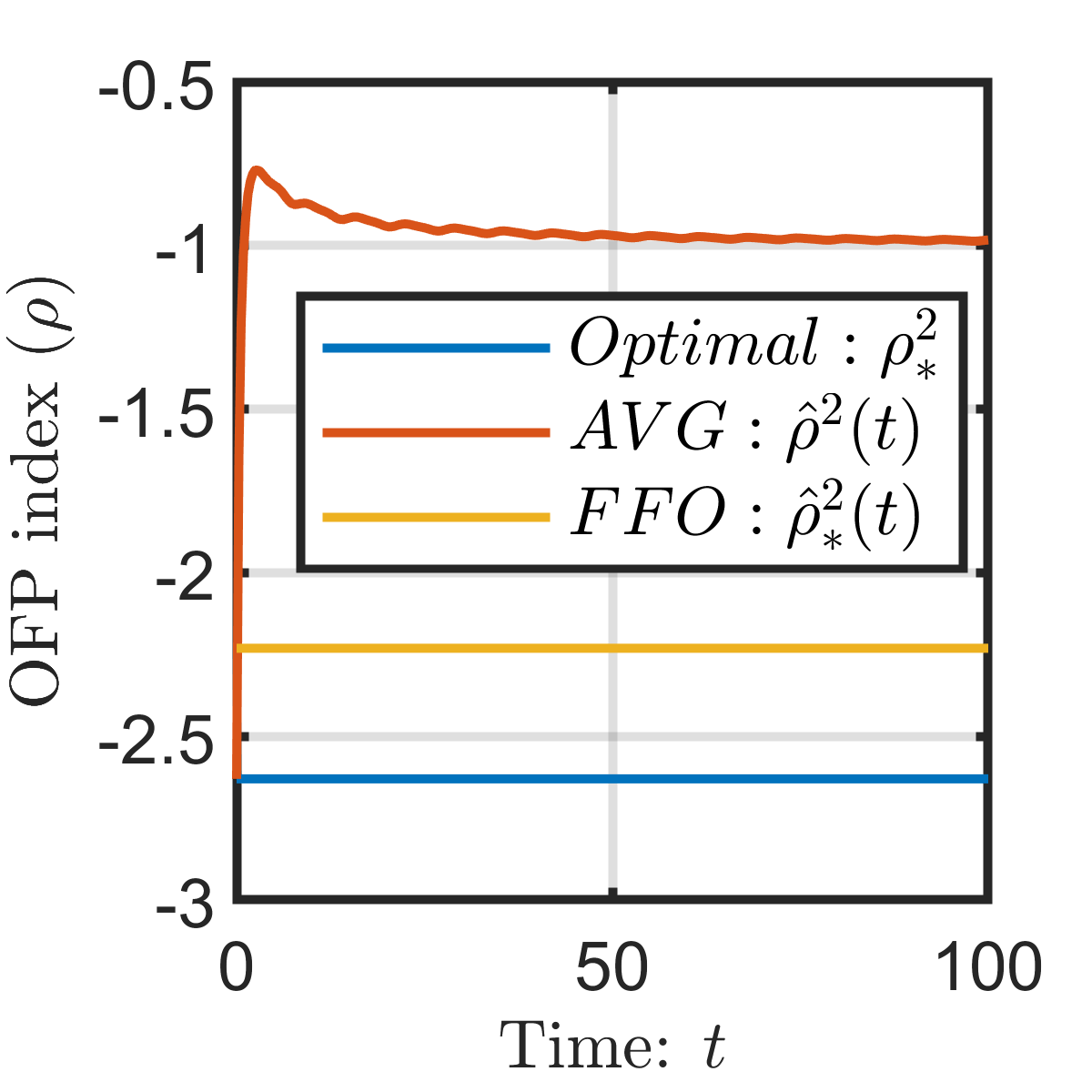}
        \caption{System: $\mathcal{H}_2$, Index: OFP}
    \end{subfigure}\\
    \begin{subfigure}[t]{0.23\textwidth}
        \centering
        \captionsetup{justification=centering}
        \includegraphics[width=1.6in]{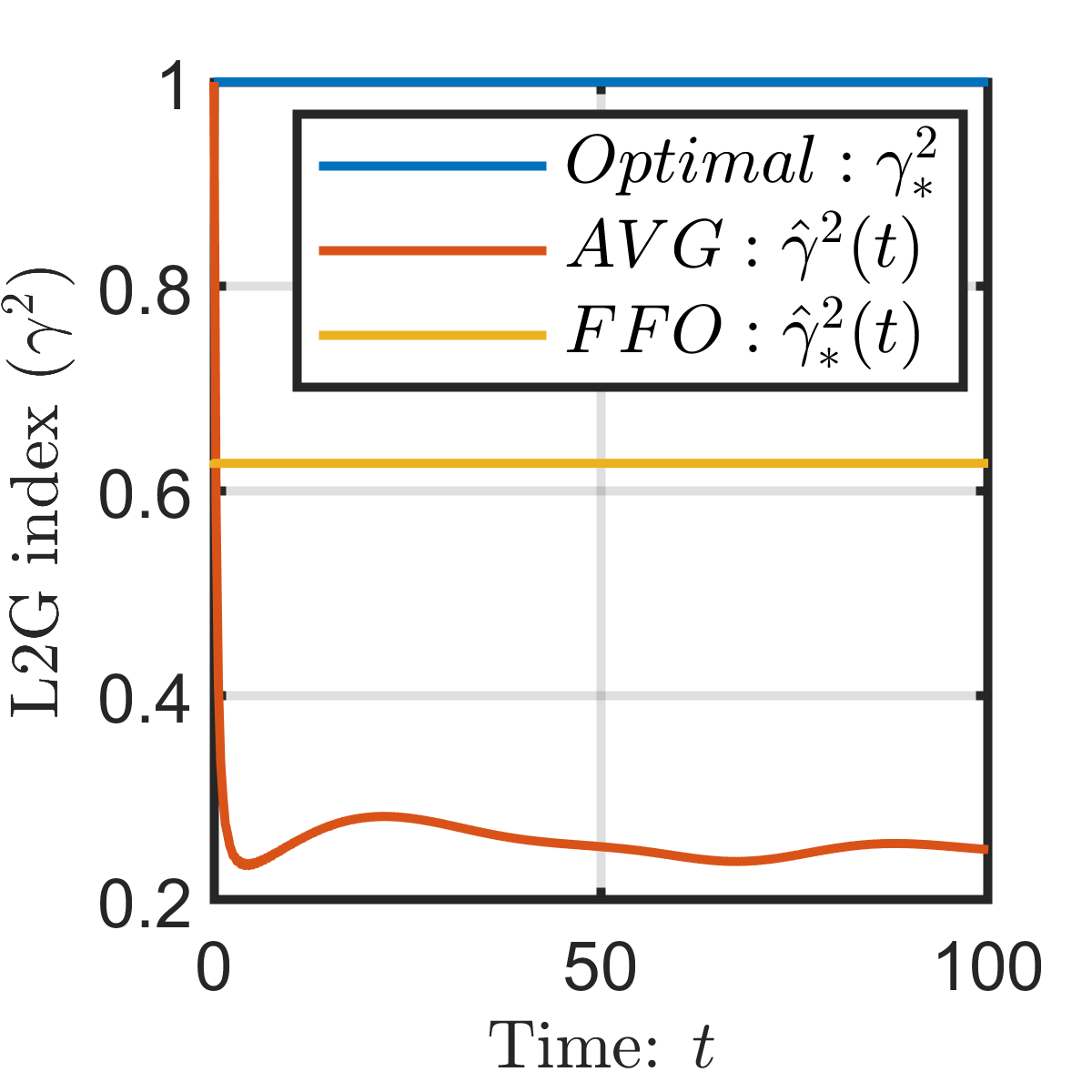}
        \caption{System: $\mathcal{H}_3$, Index: L2G}
    \end{subfigure}
    \hfill
    \begin{subfigure}[t]{0.23\textwidth}
        \centering
        \captionsetup{justification=centering}
        \includegraphics[width=1.6in]{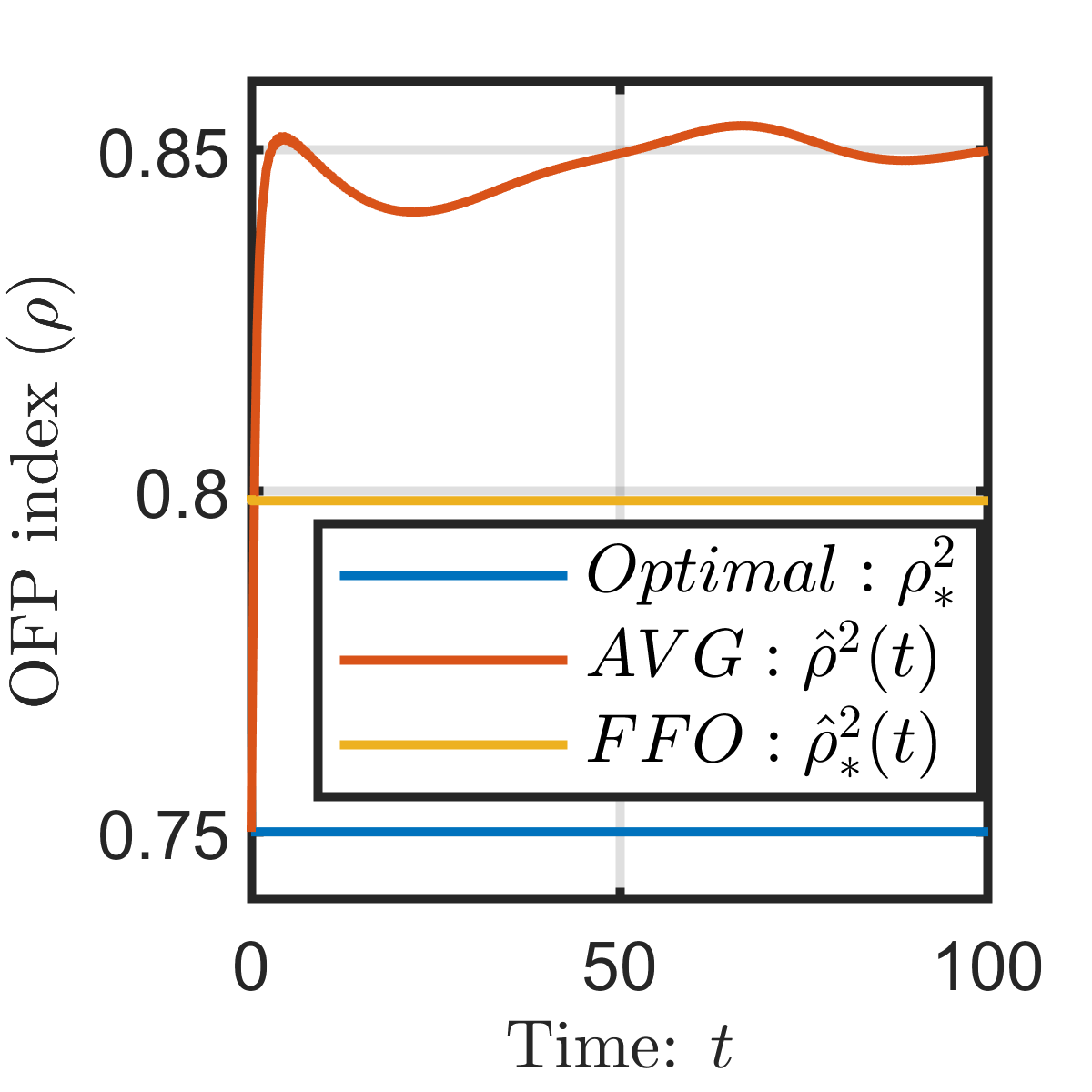}
        \caption{System: $\mathcal{H}_4$, Index: OFP}
    \end{subfigure}
    \caption{Observed on-line estimates of different interested system indices. The on-line AVG and FFO estimators respectively use \eqref{Eq:EstimatedOptimalSystemIndicesOld} and \eqref{Eq:Th:EstimatedOptimalSystemIndices} (with $t_o=0$ and $t_f = t$).}
    \label{Fig:ObservedEstimates}
\end{figure}

We define the absolute estimation error (AEE) as the absolute difference between an observed terminal (at $t=100$) estimate value of an interested system index and the corresponding optimal system index value. For example, with regard to the L2G index of a system, $\vert \gamma^2_*-\hat{\gamma}^2(100) \vert$ and $\vert \gamma_*^2 - \hat{\gamma}_*^2(100) \vert$ are the respective AEE values of AVG \cite{Zakeri2019} and FFO estimators. The terminal estimate values, the corresponding AEE values and the percentage AEE improvement achieved by the proposed on-line estimator in this paper (i.e., FFO) are shown in Tab. \ref{Tab:NumericalResults} (Cols. 5-9). From these results, it is clear that the proposed FFO estimator is significantly more accurate than the AVG estimator proposed in \cite{Zakeri2019}.  


Finally, we define the mean AEE (MAEE) as the mean estimation error of an interested system index with respect to the corresponding optimal system index value over the period $t\in [0,100)$. For example, with regard to the OFP index of a system, $\frac{1}{100} \int_0^{100}\vert \rho_*-\hat{\rho}(t)\vert dt$ and $\frac{1}{100}\int_0^{100}\vert \rho_* - \hat{\rho}_*(t) \vert dt$ are the respective MAEE values of AVG \cite{Zakeri2019} and FFO estimators. The observed MAEE values and the percentage MAEE improvement achieved by the proposed on-line estimator in this paper (i.e., FFO) are shown in Tab. \ref{Tab:NumericalResults} (Cols. 10-12). These results again confirm the significant accuracy improvement achieved by the proposed FFO estimator compared to the AVG estimator proposed in \cite{Zakeri2019}.

We conclude this paper by re-emphasizing several critical advantages of the proposed on-line estimator: 
(i) It is computationally efficient and flexible to implement (see Rm. \ref{Rm:Efficiency} and Rm. \ref{Rm:Sampling});
(ii) It provides highly convergent yet reactive/robust estimates compared those provided by the AVG estimator \cite{Zakeri2019} (see Fig. \ref{Fig:ObservedEstimates}); 
(iii) It requires significantly less amount of data compared to the existing off-line approaches such as \cite{Koch2021}; (iv) It is applicable for general non-linear systems; 
(v) It requires no control over the control input fed into the system (hence can be implemented parallel to existing control systems). The ongoing research focuses on trading off the latter two features/strengths of the proposed estimator to increase the estimator accuracy further.

\section{Conclusion}\label{Sec:Conclusion}
We proposed an on-line estimation method to estimate a few critical stability and passivity metrics (system indices) of a non-linear system using input-output data. 
First, we exploited fundamental definitions of such system indices to formulate the corresponding estimation problems as fractional function optimization problems of a specific common form. 
Next, we established several interesting theoretical results regarding this class of FFOPs. 
Then, these results were exploited to derive on-line estimates for the interested system indices. 
Compared to an existing on-line estimation method (state of the art, to the best of our knowledge), the proposed approach is more theoretically precise, computationally efficient, and enjoys several other desirable qualities. 
The validity and potential of the proposed approach were demonstrated using numerical examples. 
Future work aims to derive specific results for linear systems and optimal controls to favor on-line estimation.

\bibliographystyle{IEEEtran}
\bibliography{References}

\begin{thebibliography}{10}
\providecommand{\url}[1]{#1}
\csname url@samestyle\endcsname
\providecommand{\newblock}{\relax}
\providecommand{\bibinfo}[2]{#2}
\providecommand{\BIBentrySTDinterwordspacing}{\spaceskip=0pt\relax}
\providecommand{\BIBentryALTinterwordstretchfactor}{4}
\providecommand{\BIBentryALTinterwordspacing}{\spaceskip=\fontdimen2\font plus
\BIBentryALTinterwordstretchfactor\fontdimen3\font minus
  \fontdimen4\font\relax}
\providecommand{\BIBforeignlanguage}[2]{{%
\expandafter\ifx\csname l@#1\endcsname\relax
\typeout{** WARNING: IEEEtran.bst: No hyphenation pattern has been}%
\typeout{** loaded for the language `#1'. Using the pattern for}%
\typeout{** the default language instead.}%
\else
\language=\csname l@#1\endcsname
\fi
#2}}
\providecommand{\BIBdecl}{\relax}
\BIBdecl

\bibitem{Khalil1996}
H.~K. Khalil, \emph{{Nonlinear Systems}}.\hskip 1em plus 0.5em minus
  0.4em\relax Pearson, 2001.

\bibitem{Bao2007}
J.~Bao and P.~L. Lee, \emph{{Process control : The Passive Systems
  Approach}}.\hskip 1em plus 0.5em minus 0.4em\relax Springer, 2007.

\bibitem{Willems1972a}
J.~C. Willems, ``{Dissipative Dynamical Systems Part I: General Theory},''
  \emph{Archive for Rational Mechanics and Analysis}, vol.~45, no.~5, pp.
  321--351, 1972.

\bibitem{Zames1966}
G.~Zames, ``{On the Input-Output Stability of Time-Varying Nonlinear Feedback
  Systems Part I: Conditions Derived Using Concepts of Loop Gain, Conicity, and
  Positivity},'' \emph{IEEE Trans. on Automatic Control}, vol.~11, no.~2, pp.
  228--238, 1966.

\bibitem{Desoer1975}
C.~A. Desoer and M.~Vidyasagar, \emph{{Feedback Systems: Input-Output
  Properties}}.\hskip 1em plus 0.5em minus 0.4em\relax Academic Press, 1975.

\bibitem{Tang2019}
W.~Tang and P.~Daoutidis, ``{Input-Output Data-Driven Control Through
  Dissipativity Learning},'' in \emph{Proc. of American Control Conf.}, 2019,
  pp. 4217--4222.

\bibitem{Koch2021}
A.~Koch, J.~M. Montenbruck, and F.~Allgower, ``{Sampling Strategies for
  Data-Driven Inference of Input–Output System Properties},'' \emph{IEEE
  Trans. on Automatic Control}, vol.~66, no.~3, pp. 1144--1159, 2021.

\bibitem{Tanemura2019b}
M.~Tanemura and S.~Azuma, ``{Closed-Loop Data-Driven Estimation on Passivity
  Property},'' in \emph{Proc. of IEEE Conf. on Control Technology and
  Applications}, 2019, pp. 630--634.

\bibitem{Romer2017a}
A.~Romer, J.~M. Montenbruck, and F.~Allg{\"{o}}wer, ``{Determining Dissipation
  Inequalities From Input-Output Samples},'' in \emph{IFAC-PapersOnLine},
  vol.~50, no.~1.\hskip 1em plus 0.5em minus 0.4em\relax Elsevier, 2017, pp.
  7789--7794.

\bibitem{Romer2017b}
A.~Romer, J.~M. Montenbruck, and F.~Allgower, ``{Sampling Strategies for
  Data-Driven Inference of Passivity Properties},'' in \emph{Proc. of 56th
  Conf. on Decision and Control}, 2017, pp. 6389--6394.

\bibitem{Tanemura2019}
M.~Tanemura and S.~Azuma, ``{Efficient Data-Driven Estimation of Passivity
  Properties},'' \emph{IEEE Control Systems Letters}, vol.~3, no.~2, pp.
  398--403, 2019.

\bibitem{Iijima2020}
K.~Iijima, M.~Tanemura, S.~Azuma, and Y.~Chida, ``{Reduction in the Amount of
  Data for Data-driven Passivity Estimation},'' in \emph{Proc. of IEEE Conf. on
  Control Technology and Applications}, 2020, pp. 134--139.

\bibitem{Tang2021}
W.~Tang and P.~Daoutidis, ``{Dissipativity Learning Control (DLC): Theoretical
  Foundations of Input–Output Data-Driven Model-Free Control},''
  \emph{Systems \& Control Letters}, vol. 147, p. 104831, 2021.

\bibitem{Zakeri2019}
H.~Zakeri and P.~J. Antsaklis, ``{A Data-Driven Adaptive Controller
  Reconfiguration for Fault Mitigation: A Passivity Approach},'' in \emph{Proc.
  of 27th Mediterranean Conf. on Control and Automation}, 2019, pp. 25--30.

\bibitem{Kottenstette2014}
N.~Kottenstette, M.~J. McCourt, M.~Xia, V.~Gupta, and P.~J. Antsaklis, ``{On
  Relationships Among Passivity, Positive Realness, and Dissipativity in Linear
  Systems},'' \emph{Automatica}, vol.~50, no.~4, pp. 1003--1016, 2014.

\bibitem{Xia2018}
M.~Xia, A.~Rahnama, S.~Wang, and P.~J. Antsaklis, ``{Control Design Using
  Passivation for Stability and Performance},'' \emph{IEEE Trans. on Automatic
  Control}, vol.~63, no.~9, pp. 2987--2993, 2018.

\bibitem{Oppenheim1997}
A.~V. Oppenheim, A.~S. Willsky, and S.~H. Nawab, \emph{{Signals \&
  systems}}.\hskip 1em plus 0.5em minus 0.4em\relax Prentice Hall, 1997.

\end{thebibliography}

\end{document}